\documentclass[11pt]{article}
\usepackage[utf8]{inputenc}
\usepackage[T1]{fontenc}
\usepackage[margin=1in]{geometry}

\usepackage{amsmath,amssymb,amsthm}
\usepackage{graphicx}
\usepackage{enumitem}
\usepackage{array}
\usepackage{changepage}   % adjustwidth for the wide table
\usepackage{caption}
\usepackage{cite}
\usepackage{listings}
\usepackage[hidelinks]{hyperref}

\newlength\savedwidth

\newcommand\thickhline{\noalign{\global\savedwidth\arrayrulewidth\global\arrayrulewidth 2pt}%
  \hline
  \noalign{\global\arrayrulewidth\savedwidth}}
\newcolumntype{+}{!{\vrule width 2pt}}

\newtheorem{definition}{Definition}
\newtheorem{proposition}{Proposition}
\newtheorem{lemma}{Lemma}
\newtheorem{theorem}{Theorem}

\title{Combining model checking with simulation-based techniques for protocol verification}
\author{Takanori Ishibashi\thanks{takanori.ishibashi@jaist.ac.jp} \qquad Kazuhiro Ogata\\[4pt]
School of Information Science,\\ Japan Advanced Institute of Science and Technology,\\ Asahidai, Nomi, Ishikawa, Japan}
\date{}

\begin{document}
\maketitle

\begin{abstract}
	Model checking is a powerful technique for verifying properties of systems or protocols, but it often suffers from the state space explosion problem.
	For protocols such as the Alternating Bit Protocol (ABP) and the Sliding Window Protocol (SWP), increasing parameters such as channel capacities or window sizes leads to a reachable state space that is infeasible for exhaustive model checking.
	To address this, we propose a technique that combines model checking with simulation relations.
	We focus on three protocols: the Simple Communication Protocol (SCP), ABP, and SWP.
	Formalized as state machines, these protocols are presented in decreasing order of abstraction: SCP, ABP, and then SWP, which results in a corresponding increase in reachable states.
	While SCP is verifiable through direct model checking without state space explosion, the direct application of model checking to ABP and SWP becomes infeasible as their parameter values increase.
	We demonstrate that ABP satisfies its invariant property by combining model checking on SCP with a simulation relation from ABP to SCP.
	Furthermore, we demonstrate that SWP satisfies its invariant property through the composition of simulation relations from SWP to ABP and from ABP to SCP, in conjunction with model checking on SCP.
	These approaches enable the formal verification of protocols with large reachable state spaces by conducting model checking on a significantly  smaller state machine.
\end{abstract}

\section{Introduction}\label{sec:introduction}
Model checking \cite{Clarke2008} is a significant achievement in computer science, widely used in the hardware and software industries.
It is very important to verify the correctness of hardware and software systems to ensure their reliability. Model checking is one of the most advanced techniques for verifying that a finite-state system satisfies its desired properties.
However, model checking has several challenges, of which state space explosion is the most notorious.
It often makes it infeasible to conduct model checking experiments.
Many researchers have addressed this problem and developed a variety of techniques, such as abstraction \cite{edmund1, edmund2, jose1, yin1, yin2, bae, modelchecking01, modelchecking02} and partial order reduction \cite{Clarke1999, modelchecking01, DBLP:journals/pacmpl/EneaGKM24}.
Despite these efforts, the challenge remains.
It is therefore a critical area for further research.
While existing automatic abstraction techniques provide some relief, they often face limitations when dealing with complex protocol parameters.
To address this, we propose a methodology that leverages domain knowledge to construct simulation relations, enabling more scalable verification.

In this paper, we investigate this issue using three protocols: the Simple Communication Protocol (SCP) \cite{ogata2008simulation}, the Alternating Bit Protocol (ABP) \cite{abp}, and the Sliding Window Protocol (SWP) \cite{tanenbaum, veeraswamy}.
Detailed informal descriptions of these protocols are omitted for brevity; readers are referred to \cite{ogata2008simulation, abp, tanenbaum, veeraswamy} for further details.
These protocols are characterized by different levels of complexity: SCP has the highest degree of abstraction and the smallest number of reachable states, while SWP has the lowest degree of abstraction and the largest number of reachable states.
ABP serves as an intermediate protocol between SCP and SWP in terms of both abstraction and the number of reachable states.
These protocols are selected as a systematic benchmark suite to demonstrate the scalability of our approach.
Direct model checking is feasible for SCP but becomes infeasible for ABP and SWP due to the state space explosion problem as their parameter values increase.
To overcome this limitation, we propose a technique that combines model checking with simulation relations.
A simulation relation ensures that if invariant properties hold on an abstract state machine, the corresponding invariant properties are preserved on the more concrete machine.
In this paper, we focus on invariant properties.
Invariant properties are state predicates that hold true in each initial state and remain true throughout any actions (or state transitions) of systems or protocols, remaining unchanged regardless of their reachable states (states reachable from an initial state with zero or more state transitions).
We first propose a verification methodology that combines model checking with simulation relations, enabling the formal verification of protocols that are infeasible for direct model checking due to state space explosion.
We then extend this approach by introducing a methodology for composing simulation relations, which allows for the formal verification of even larger and more complex models through hierarchical abstraction.
Based on these verification methodologies, the main contributions of this study are as follows:
\begin{itemize}
	\item A methodology that combines model checking with simulation relations to enable the verification of systems or protocols where direct model checking fails due to the state space explosion problem.
	\item A methodology for the composition of simulation relations, which further extends verification capabilities to larger models by linking multiple layers of abstraction.
\end{itemize}

All formal specifications and the Maude commands used for verification are available online at \url{https://zenodo.org/records/19734839}.
The primary objective of this study is to determine whether our proposed method can effectively mitigate the state space explosion problem.
To evaluate this, we address the following two research questions:

\begin{itemize}
	\item \textbf{RQ1:} Can the problem of being unable to verify that ABP satisfies its invariant property due to state space explosion be resolved by combining model checking with the simulation relation between ABP and SCP?
	\item \textbf{RQ2:} Can the problem of being unable to verify that SWP satisfies its invariant property due to state space explosion be resolved through the composition of simulation relations, specifically the simulation relations from SWP to ABP and from ABP to SCP?
\end{itemize}

By answering these questions, we show that our technique enables the formal verification of complex protocols with large reachable state spaces that are otherwise beyond the reach of direct model checking.
The remaining part of the paper is organized as follows.
Section~\ref{sec:preliminaries} provides the foundational concepts required for this study, including state machines, simulation relations, and the Maude specification language.
We use the Bare Communication Protocol (BCP) \cite{ogata2008simulation} as an example to demonstrate the formalization of protocols in Maude and the application of its commands.
Section~\ref{sec:proposed_verification_methodologies} presents our two proposed theoretical frameworks: a methodology that combines model checking with simulation relations and a methodology for the composition of simulation relations.
Section~\ref{sec:case_study_1} verifies the invariant properties of ABP by applying the methodology that combines model checking with simulation relations.
Section~\ref{sec:case_study_2} verifies the invariant properties of SWP by employing the composition of simulation relations to establish a hierarchical abstraction.
Section~\ref{sec:related_work} mentions some existing related work.
Section~\ref{sec:conclusion} concludes the paper including some future directions.

\section{Preliminaries}\label{sec:preliminaries}
\subsection{State machines}
\begin{definition}[State Machine]
	A state machine $M$ is $\langle S_{M}, I_{M}, T_{M} \rangle$ such that
	\begin{itemize}
		\item $S_{M}$ is a set of states.
		\item $I_{M} \subseteq S_{M}$ is the set of initial states.
		\item $T_{M} \subseteq S_{M} \times S_{M}$ is a binary relation over states. $(s, s') \in T_{M}$ may be written as $s \leadsto_{M} s'$, called a state transition, where $s'$ is called a successor state of $s$.
	\end{itemize}
\end{definition}
For simplicity of notation, the subscript $_{M}$ is omitted if it is clear from the context that a symbol refers to the state machine $M$.
$\leadsto^*_{M}$ is the reflexive transitive closure of $\leadsto_{M}$.

\begin{definition}[Reachable States]
	\label{reachable_states}
	The set $R_{M}$ of reachable states with respect to (wrt) $M$ is defined as the smallest set satisfying the following conditions:
	\begin{enumerate}[label=(\roman*)]
		\item $I_{M} \subseteq R_{M}$
		\item If $s \in R_{M}$ and $s \leadsto_{M} s'$, then $s' \in R_{M}$
	\end{enumerate}
\end{definition}

\begin{definition}[Invariant Property]
	Let $M$ be a state machine.
	A state predicate $p$ is an invariant property wrt $M$ if $p(s)$ is true for all $s \in  R_{M}$, i.e. $(\forall s \in R_{M})\centerdot p(s)$.
\end{definition}
That a state machine satisfies an invariant property is equivalent to saying that a state predicate is invariant wrt a state machine.
A state predicate $p$ can be interpreted as a set $P_{M}$ of states such that $(\forall s \in P_{M})\centerdot p(s)$ and $(\forall s \notin P_{M})\centerdot \neg p(s)$.

\subsection{Simulations}
A simulation denotes a relation where one program faithfully executes the computational steps of another, maintaining the same logical sequence despite potential differences \cite{DBLP:conf/ijcai/Milner71}.
The paper \cite{ogata2008simulation} defines the simulation relation of one state machine to another state machine.
A type of state machine called Observable Transition Systems (OTS) is used for definitions;
however, in this paper we use the state machines defined in the previous subsection, rather than OTS.

\begin{definition}[Simulation Relation]
	\label{definition_simulation}
	Let $M$ and $M_A$ be state machines. A relation $r : R_M \times R_{M_A} \to \mathrm{Bool}$ is called a simulation relation from $M$ to $M_A$ if it satisfies the following conditions:
	\begin{enumerate}[label=(\roman*)]
		\item For each $s \in I_{M}$, there exists $s_A \in I_{M_A}$ such that $r(s, s_A)$.
		\item For each $s, s ' \in R_M$ and $s_A \in R_{M_A}$ such that $r(s, s_A)$ and $s \leadsto_{M} s'$, there exists $s_A ' \in R_{M_A}$ such that $r(s ', s_A ')$ and $s_A \leadsto^*_{M_A} s'_A$.
	\end{enumerate}
	Note that if $s_A \in R_{M_A}$ and $s_A \leadsto^*_{M_A} s'_A$,
	then $s_A ' \in R_{M_A}$.
\end{definition}

Fig \ref{fig_simulation} shows the diagrams corresponding to the two conditions in Definition \ref{definition_simulation}.
When there exists a simulation relation $r$ from $M$ to $M_A$,
we may say that $M$ is simulated by $M_A$, or $M_A$ simulates $M$.
Definition \ref{definition_simulation} does not say explicitly that for each $s \in R_M$,
there exists $s_A \in R_{M_A}$ such that $r(s, s_A)$,
and then we need to have the following proposition:
\begin{proposition}
	\label{proposition1}
	Let $M$ and $M_A$ be state machines.
	Suppose that $M$ is simulated by $M_A$, that is, there exists a simulation relation from $M$ to $M_A$.
	Then, for every such simulation relation $r$ from $M$ to $M_A$ and every state $s \in R_M$, there exists a state $s_A \in R_{M_A}$ such that $r(s, s_A)$.
\end{proposition}
\begin{proof}
	By induction on $s \in R_{M}$.
	If $s \in I_{M}$, then there exists $s_A \in R_{M_A}$ such that $r(s, s_A)$ from Definition \ref{definition_simulation} (i).
	If $s \leadsto_{M} s'$, then there exists $s_A \in R_{M_A}$ such that $r(s, s_A)$ from the induction hypothesis.
	Then, from Definition \ref{definition_simulation} (ii), there exists $s'_A \in R_{M_A}$ such that $r(s', s'_A)$ and $s_A \leadsto^*_{M_A} s'_A$.
\end{proof}

\begin{figure}[htbp]
	\centering
	\includegraphics[width=0.7\linewidth]{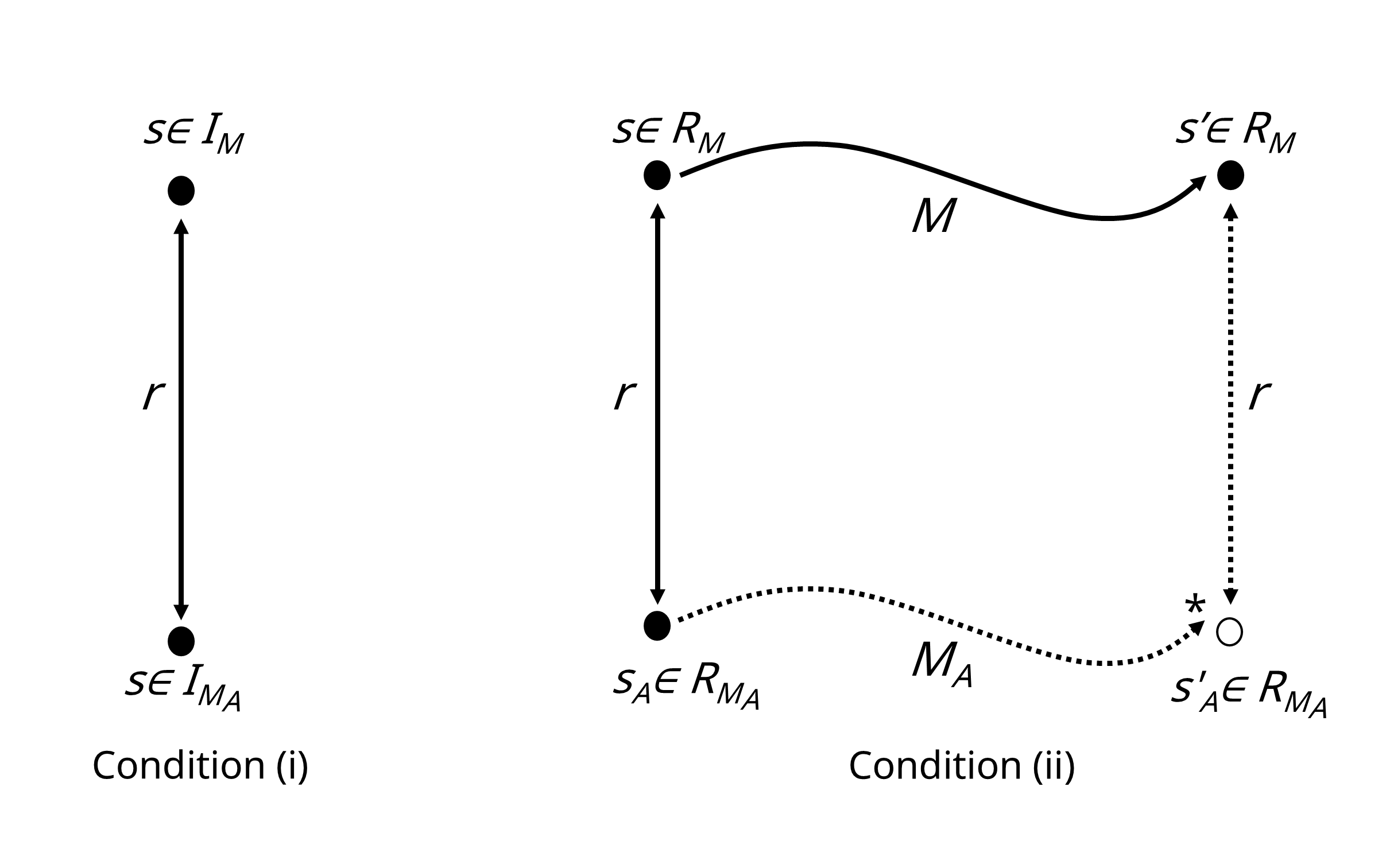}
	\caption{\textbf{A simulation relation $r$ from $M$ to $M_A$.}}
	\label{fig_simulation}
\end{figure}

If $M$ is simulated by $M_A$, a state predicate $p_A$ is an invariant property wrt $M_A$ and $p_A$ is preserved by the simulation relation, then the counterpart state predicate $p$ is an invariant property wrt $M$.

\begin{theorem}
	\label{theorem1}
	Let $M$ and $M_A$ be state machines, and let $p$ and $p_A$ be state predicates.
	Suppose that $M$ is simulated by $M_A$.
	Then, for every such simulation relation $r$ from $M$ to $M_A$, if $p_A(s_A) \Rightarrow p(s)$ holds for any states $s \in R_M$ and $s_A \in R_{M_A}$ satisfying $r(s, s_A)$, the following holds:
	if $p_A$ is an invariant property wrt $M_A$, then $p$ is an invariant property wrt $M$.
\end{theorem}
\begin{proof}
	For an arbitrary $s \in R_{M}$,
	there exists $s_A \in R_{M_A}$ with $r(s, s_A)$ from Proposition \ref{proposition1}.
	Therefore, $p(s)$ holds from the assumption.
\end{proof}

\subsection{Maude}
Maude \cite{maude} is a specification and programming language based on rewriting logic.
In Maude, functional modules (\texttt{fmod}) define the data of the target system, such as sorts (data types), operators, and equations.
In contrast, system modules  (\texttt{mod}) define the system's transition model by allowing rewrite rules in addition to the declarations used in functional modules.
States can be represented in various forms.
In model checking using search command of Maude, we describe a state as a braced associative-commutative collection of name-value pairs, where each name may include parameters. Within the Maude community's terminology, associative-commutative collections are referred to as ``soup'', and name-value pairs are referred to as ``observable components''.
Therefore, a state is expressed as a braced soup of observable components.
The juxtaposition operator serves as the constructor for soups.
Let \texttt{$oc_0, \, oc_1, \, oc_2$} be observable components; then \texttt{$oc_0 \, oc_1 \, oc_2$} constructs the soup containing these three observable components.
Consequently, a state is expressed as \texttt{$\{ oc_0 \, oc_1 \, oc_2 \}$}.
This associative-commutative nature allows Maude to represent complex state structures flexibly.
While Maude's internal engine handles these structures automatically during standard model checking, the multiple substitutions resulting from such matching are explicitly addressed in our simulation-based approach.
This necessitates the systematic case splitting discussed in Section~\ref{sec:proposed_verification_methodologies}.

The state transition relation \textit{T} is defined using rewrite rules.
A standard rewrite rule begins with the keyword \texttt{rl}, followed by a label enclosed in square brackets, a colon, two patterns separated by \texttt{=>}, and concludes with a period. In contrast, a conditional rewrite rule, initiated with \texttt{crl}, includes a condition following the \texttt{if} keyword and ends with a period.
An example of a conditional rewrite rule is
\[
	\texttt{crl [$lb$] : $t$ => $t'$ if $c_1$ $/\backslash$ $\dots$ $/\backslash$ $c_n$ .}
\]
where $lb$ is the rule's label, and each $c_i$ (for $i = 1, \dots, n$) is a condition, often represented as an equation \texttt{$tc_i$ = ${tc_i}'$}.
The negation of this equation can be expressed as \texttt{($tc_i$ != ${tc_i}'$) = true}, with the \texttt{= true} part being optional. If the condition \textit{$c_1$ $/\backslash$ $\dots$ $/\backslash$ $c_n$} is satisfied under a certain substitution $\sigma$, then $\sigma(t)$ can be replaced with $\sigma(t')$.
The substitution $\sigma$ is a function from variables to terms that preserves sorts.
The substitution $\sigma$ can be naturally extended as a function from terms to terms as follows:
\[
	\sigma\bigl(f(t_1,\ldots,t_n)\bigr)
	= f\bigl(\sigma(t_1),\ldots,\sigma(t_n)\bigr)
\]
for any function symbol $f$ of arity $n$ and terms $t_1,\ldots,t_n$.

Maude provides lists as one of its predefined container data types.
The lists are parameterized, allowing for the construction of lists of elements of any data type.
The basic representation of lists uses two free constructors: the constant \texttt{nil} for the empty list and the juxtaposition operator for list construction.
For example, when an element \texttt{E} of the same data type as the elements of list \texttt{L} is added to the front of the list, we denote it as \texttt{E L}.

Maude makes it possible to use fresh constants that represent arbitrary values, thereby enabling symbolic execution.
For example, a fresh constant $x$ of sort \texttt{Nat} represents an arbitrary natural number.

We describe how to specify a state machine in Maude through the Bare Communication Protocol (BCP) \cite{ogata2008simulation} as an example.
To formalize the BCP as a state machine, we use the following observable components.
\begin{itemize}
	\item \texttt{(index: i)} - \texttt{i} is a number. The sender transmits this number to the receiver.
	\item \texttt{(list: l)} - \texttt{l} is a list of numbers. The receiver receives numbers from the sender and stores them in this list.
\end{itemize}

The initial state \texttt{init} is defined as follows:
\begin{lstlisting}
{(index: 0) (list: nil)} .
\end{lstlisting}

In the initial state, each value is as follows:
\begin{itemize}
	\item The \texttt{index} number is set to 0. This indicates that the sender is prepared to send a message with the index number 0.
	\item The \texttt{list}, which the receiver uses to store index numbers received from the sender, is empty (\texttt{nil}). This indicates that the receiver has not yet received any index numbers.
\end{itemize}

The BCP is specified by a single rewrite rule \texttt{send}.
The rewrite rules use the following Maude variables:
\begin{itemize}
	\item \texttt{OCs} denotes a variable of observable components soups.
	\item \texttt{I} denotes a variable of \texttt{index}.
	\item \texttt{L} denotes a variable of \texttt{list}.
\end{itemize}
The rewrite rule \texttt{send} is defined as follows:
\begin{lstlisting}
crl [send] : 
{(index: I) (list: L) OCs} => {(index: s(I)) (list: (I L)) OCs}
if I <= maxIndex .
\end{lstlisting}
The rewrite rule specifies that if the index number, denoted as \texttt{I}, is less than or equal to the predefined \texttt{maxIndex},
the index is incremented by 1 (\texttt{s(I)})
and the index is added to the list (\texttt{(I L)}).
\texttt{maxIndex} is a natural number.
It is used to limit the reachable state space during model checking.

The invariant property candidate of BCP is defined in Maude as follows:
\begin{lstlisting}
eq BCP-inv(BCP) = mk(getIndex(BCP)) = (getIndex(BCP) getList(BCP)) .
\end{lstlisting}
\texttt{BCP-inv} ensures that the protocol's state remains consistent and correct.
Letting \texttt{BCP} denote the protocol's state,
\texttt{BCP-inv} is defined such that \texttt{mk(getIndex(BCP))} is equal to \texttt{(getIndex(BCP) getList(BCP))}.
\texttt{mk} takes a natural number \textit{k} and makes the list $(k \ k-1 \ \ldots \ 0)$.
The functions prefixed with \texttt{get}, including \texttt{getIndex} and \texttt{getList}, serve as accessors to extract the \texttt{index} and \texttt{list} components, respectively, from the state.
We define a constructor \texttt{bcp} for the data type \texttt{BCPState},
which takes two arguments for \texttt{index} and \texttt{list},
to pass the protocol state as an argument to functions evaluating invariant property candidates and simulation relation candidates.
For instance, \texttt{getList(bcp(index, list))} returns \texttt{list}.
For other communication protocols in this paper, we define states in the same manner.

For model checking, Maude provides the \texttt{search} command, which looks for states that are reachable from an initial state through state transitions, aiming to find states that meet a certain condition.
To verify that BCP satisfies its invariant property \texttt{BCP-inv}, the following command is executed:
\begin{lstlisting}
search[1] in BCP : init =>* {(index: I) (list: L) OCs} 
such that not BCP-inv(bcp(I, L)) .
\end{lstlisting}
Maude returns \texttt{No solution}.
This means that the invariant property holds throughout all reachable states of the protocol.

The \texttt{reduce} (or \texttt{red}) command is used for term evaluation, simplifying a specified term using the function defined in a module.
We will use the function \texttt{mk}, which was used within the function known as \texttt{BCP-inv}, as a concrete example.
We define the function  \texttt{mk} in the module \texttt{N-LIST}, which implements lists of natural numbers, as follows:
\begin{lstlisting}
eq mk(0) = 0 nil .
eq mk(s(X)) = s(X) mk(X) .
\end{lstlisting}
\texttt{s} represents the successor function of natural numbers.
We use Maude to reduce \texttt{mk(3)} by using the following \texttt{red} command:
\begin{lstlisting}
red in N-LIST : mk(3) .
\end{lstlisting}
Maude returns \texttt{3 2 1 0}.
This means that \texttt{mk(3)} is reduced to \texttt{3 2 1 0}.

Finally, the \texttt{srewrite} (or \texttt{srew}) command identifies successor states by applying a specific rewrite rule to a given state.
We use following the \texttt{srew} command, as a concrete example.
\begin{lstlisting}
srew in BCP : init using send .
\end{lstlisting}
This \texttt{srew} command applies the rewrite rule \texttt{send} to the initial state \texttt{init}.
The \texttt{srew} command identifies \lstinline|{index: 1 list: 0}| as the successor state.
This means that the initial state \texttt{init} moves to the state \lstinline|{index: 1 list: 0}| by an application of the rewrite rule \texttt{send} to the initial state.

\section{Proposed verification methodologies}\label{sec:proposed_verification_methodologies}

\subsection{Combining model checking with simulation relations}\label{sec:combining_model_checking_with_simulation_relations}
The methodology combining model checking with simulation relations employs model checking to demonstrate that a state machine $M$ is simulated by a more abstract state machine $M_A$.
It further shows that, if $M$ is simulated by $M_A$ and a state predicate $p_A$ is an invariant property wrt $M_A$, then the corresponding state predicate $p$ is an invariant property wrt $M$.
To prove that $p$ is an invariant property wrt $M$ using the methodology for the state machines $M$ and $M_A$, and their respective state predicates $p$ and $p_A$, we follow three steps:
\begin{enumerate}
	\item Conduct model checking to verify that $p_A$ is an invariant property wrt $M_A$.
	\item Conjecture a simulation relation candidate $r$ from $M$ to $M_A$, and prove that $p$ can be deduced from $p_A$ assuming the simulation relation candidate $r$ using static case splitting on data structures and reduction.
	\item Demonstrate that $M$ is simulated by $M_A$ using transition-based case splitting over rewrite rules, model checking, and reduction.
\end{enumerate}
Fig \ref{overview} provides a comprehensive overview of the proposed methodology.
The figure illustrates the process of verifying an invariant property $p$ of $M$ by leveraging $M_A$ and a simulation relation $r$.
As shown in the workflow, the verification is structured into three steps.
By following these steps, we can formally verify that $p$ is an invariant property wrt $M$ even though it is infeasible to directly model check that $p$ is an invariant property wrt $M$.
This is useful when $M$ has a large reachable state space that leads to state space explosion.
$M_A$ must have a reachable state space that does not lead to state space explosion.
If we successfully carry out steps 1 through 3, then we can conclude that $p$ is an invariant property wrt $M$ from Theorem \ref{theorem1}.
In the case studies of this paper, $M_A$ is manually constructed and is not automatically generated from $M$.

\begin{figure}[htbp]
	\centering
	\includegraphics[width=\linewidth]{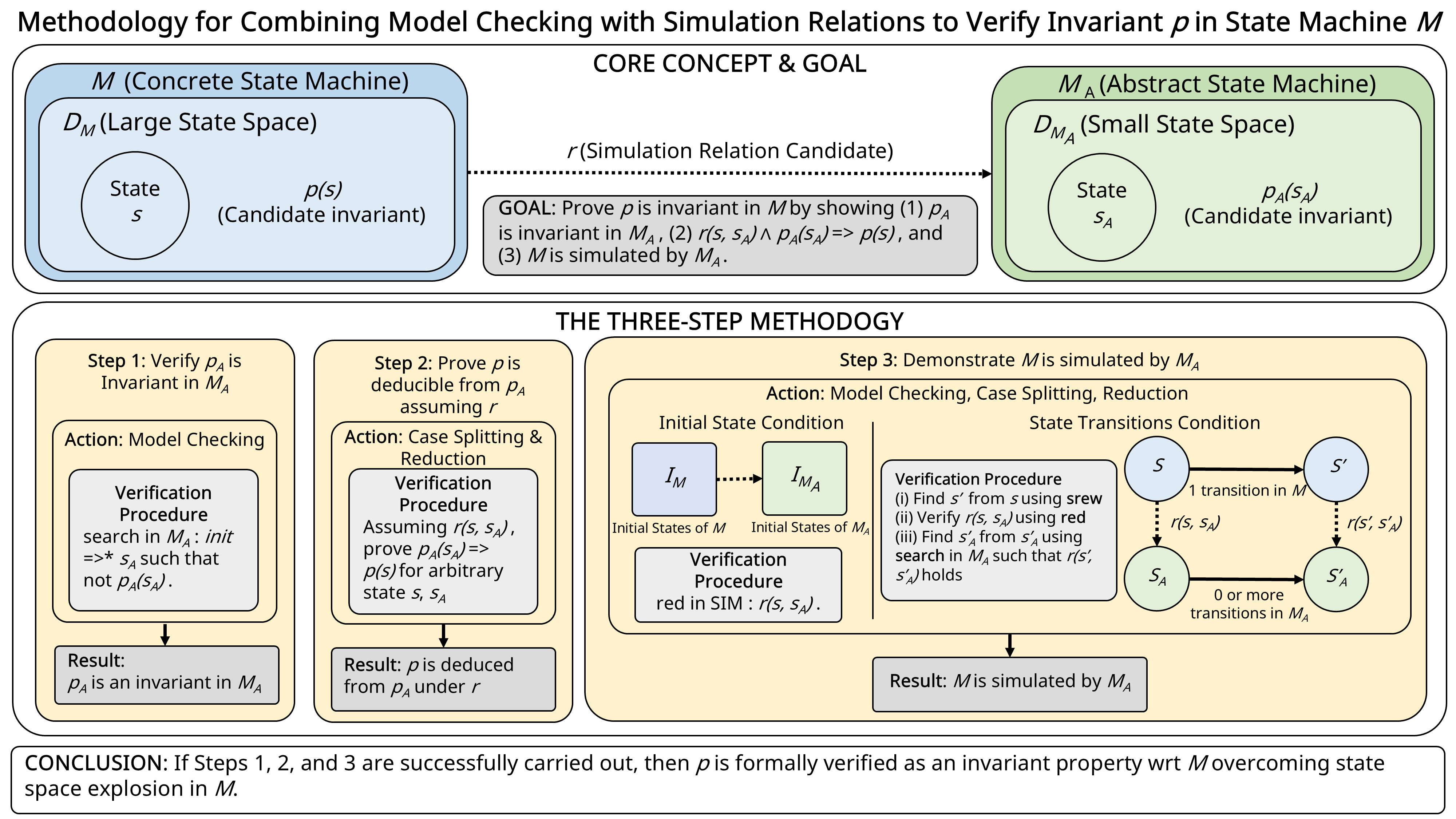}
	\caption{\textbf{Overview of the proposed methodology for combining model checking with simulation relations to verify invariant properties.}}
	\label{overview}
\end{figure}

We design the technique in step 3 to satisfy Definition \ref{definition_simulation}.
We define the simulation relation candidate $r$ with an equation in Maude.
$r$ takes two states of $M$ and $M_{A}$ respectively as arguments and returns a Boolean value.
It returns true if $r$ holds between the two states, and false otherwise.

To demonstrate that Definition \ref{definition_simulation} (i) holds, we verify that for each $s \in I_M$, there exists $s_A \in I_{M_A}$ such that $r(s, s_A)$.
To do this, we use the \texttt{red} command of Maude.
Specifically, human users are supposed to construct an initial state $s_A$ of $M_A$ or choose $s_A$ from $I_{A}$ for each $s$ of $I$ and use the following command:
\begin{lstlisting}
red in SIM : r((*@$s$@*), (*@$s_A$@*)).
\end{lstlisting}
where \texttt{SIM} is the module in which \texttt{r} is defined and then the formal specifications of $M$ and $M_{A}$ in Maude are available.
Note that it is straightforward to construct an initial state $s_A$ of $M_{A}$ or choose $s_A$ from $I_{A}$ for each $s$ of $I$.
This command returns a Boolean value.
For each $s \in I_M$, if this command returns true,
we can conclude that Definition \ref{definition_simulation} (i) is satisfied.
Assuming that $M_A$ has a reachable state space that does not lead to state space explosion, $I_{A}$ consists of a small number of initial states.
It is usually very small.
For example, all case studies in this paper involve state machines with exactly one initial state.

Next, we describe the technique for demonstrating that Definition \ref{definition_simulation} (ii) holds.
This technique requires transition-based case splitting for each state transition of $M$.
Ordinary finite-state model checking usually uses specific concrete values.
However, since we aim to demonstrate that the simulation relation is preserved across all state transitions of $M$ with a large reachable state space, where state space explosion occurs,
we use fresh constants in addition to concrete values.

We verify that for each state transition of $M$,
for arbitrary states $s, s'$ of $M$, and an arbitrary state $s_A$ of $M_{A}$ such that $r(s, s_A)$ and $s \leadsto_{M} s'$,
there exists a state $s_A '$ of $M_{A}$ such that $r(s ', s_A ')$ and $s_A \leadsto^*_{M_A} s'_A$.
In Maude, state transitions are defined by rewrite rules.
As described in the Section~\ref{sec:preliminaries}, rewrite rules are defined in two ways: as a rewrite rule (\texttt{rl}) and a conditional rewrite rule (\texttt{crl}).
\begin{itemize}
	\item \texttt{rl [$lb$] : $t$ => $t'$ }
	\item \texttt{crl [$lb$] : $t$ => $t'$ if $\ldots$ $/\backslash$ $c_i$ $/\backslash$ $\ldots$}
\end{itemize}
If the \textit{t'} in \texttt{rl} or \texttt{crl} contains an \texttt{if\_then\_else\_fi} expression, a conditional branch occurs.
Depending on the condition, different state transitions may occur.
We need to do case splitting for (conditional) rewrite rules of $M$.
In the case studies of this paper, we use case splitting based on the conditions of the conditional rewrite rules and the conditions specified in the \texttt{if\_then\_else\_fi} expression, ensuring that all possible state transitions of $M$ can be taken into account.
In addition to these cases, other types of case splitting may also be required.
For example, when using associative commutative pattern matching, multiple substitutions may be obtained, requiring further case splitting based on these substitutions.
An arbitrary state before a state transition, caused by a single application of a rewrite rule, corresponds to $\sigma(t)$,
and an arbitrary state after the state transition, caused by the same application of the same rewrite rule, corresponds to $\sigma(t')$, where $\sigma(t), \sigma(t') \in S_M$ and $\sigma(t) \leadsto_{M} \sigma(t')$.
Since reachability is, in general, undecidable, it is not feasible to restrict the proof directly to the set of reachable states $R_M$.
Therefore, we introduce a superset $D_M$ such that $R_M \subseteq D_M \subseteq S_M$.
Correspondingly, we introduce a superset $D_{M_A}$ such that $R_{M_A} \subseteq D_{M_A} \subseteq S_{M_A}$ to serve as the counterpart.
We carry out the proof over $D_M$ and $D_{M_A}$.
To ensure the soundness of this approach,
we establish that verifying the state transition condition over these supersets ($D_M$ and $D_{M_A}$) serves as a sufficient condition for satisfying Definition \ref{definition_simulation} (ii) over the original reachable states ($R_M$ and $R_{M_A}$).
This theoretical foundation is formally provided by the following lemma:
\begin{lemma}
	\label{lemma_model_checking_for_simulation}
	The logical formula
	\begin{equation}
		\label{eq:first_formula}
		\begin{aligned}
			 & (\forall s \in D_M)(\forall s' \in D_M)(\forall s_A \in D_{M_A})(\exists s_A' \in D_{M_A}) \\
			 & \quad (r(s, s_A) \wedge (s, s') \in T_M)
			\;\Rightarrow\;
			(r(s', s'_A) \wedge s_A \leadsto^*_{M_A} s'_A)
		\end{aligned}
	\end{equation}
	is a sufficient condition for the logical formula
	\begin{equation}
		\label{eq:second_formula}
		\begin{aligned}
			 & (\forall s \in R_M)(\forall s' \in R_M)(\forall s_A \in R_{M_A})(\exists s_A' \in R_{M_A}) \\
			 & \quad (r(s, s_A) \wedge (s, s') \in T_M)
			\;\Rightarrow\;
			(r(s', s'_A) \wedge s_A \leadsto^*_{M_A} s'_A) .
		\end{aligned}
	\end{equation}
\end{lemma}
\begin{proof}
	Assume that the logical formula \eqref{eq:first_formula} holds.
	Since $R_M \subseteq D_M$ and $R_{M_A} \subseteq D_{M_A}$, the following holds.
	\begin{equation}
		\tag{3}
		\label{eq:third_formula}
		\begin{aligned}
			 & (\forall s \in R_M)(\forall s' \in R_M)(\forall s_A \in R_{M_A}) (\exists s_A' \in D_{M_A}) \\
			 & \quad (r(s, s_A) \wedge (s, s') \in T_M)
			\;\Rightarrow\;
			(r(s', s'_A) \wedge s_A \leadsto^*_{M_A} s'_A)
		\end{aligned}
	\end{equation}
	We show that the logical formula \eqref{eq:third_formula} implies the logical formula \eqref{eq:second_formula}.
	\begin{enumerate}[label=(\roman*)]
		\item
		      If $(r(s, s_A) \wedge (s, s') \in T_M )$ is false,
		      the implication $(r(s, s_A) \wedge (s, s') \in T_M) \Rightarrow (r(s', s'_A) \wedge s_A \leadsto^*_{M_A} s'_A)$ is satisfied.
		      The logical formula \eqref{eq:second_formula} holds.
		\item
		      If $(r(s, s_A) \wedge (s, s') \in T_M)$ is true,
		      there exists \(s'_A \in D_{M_A}\) such that $r(s', s'_A) \wedge s_A \leadsto^*_{M_A} s'_A$.
		      Since \(s_A \in R_{M_A}\) and \(s_A \leadsto^*_{M_A} s'_A\), it follows that \(s'_A \in R_{M_A}\) from Definition \ref{reachable_states}.
		      The logical formula \eqref{eq:second_formula} holds.
	\end{enumerate}
	Hence, in both cases, the logical formula \eqref{eq:second_formula} is satisfied whenever the logical formula \eqref{eq:first_formula} holds.
	Therefore, The logical formula \eqref{eq:first_formula} is a sufficient condition for the logical formula \eqref{eq:second_formula}.
\end{proof}
This lemma allows us to reason soundly about reachable states by performing verification over their respective supersets.
To implement this verification in Maude, we execute a three-command procedure for each rewrite rule of $M$.

First, we find the states resulting from the state transition caused by the application of the rewrite rule using the \texttt{srew} command of Maude.
Specifically, we use the following command:
\begin{lstlisting}
srew in (*@$M$@*) : (*@$\sigma(t)$@*) using (*@$lb$@*) .
\end{lstlisting}
where $M$ is the name of the module defining the state machine $M$, and $lb$ is the label of the rewrite rule.
This \texttt{srew} command finds $\sigma(t')$ that is reachable from $\sigma(t)$ in $M$ by applying the rewrite rule \textit{lb}.
$\sigma(t)$ is the term that represents the most general state for each case in the case splitting.
We construct $\sigma(t)$ using fresh constants.
By performing the exhaustive case splitting, we ensure that the collection of these terms $\sigma(t)$ effectively covers the entire state space, representing the most general state possible for every transition branch.
This comprehensive coverage is essential to guarantee that no potential state transition is overlooked during the verification process.
For example, if $t$ has associativity or commutativity, multiple substitutions can be obtained.
It can be observed that when multiple substitutions are obtained, further case splitting is required.
The point at which the \texttt{srew} command is applied is shown as command (i) in Fig \ref{3steps}.
This command helps to find $\sigma(t')$ required to demonstrate Definition \ref{definition_simulation} (ii).

\begin{figure}[htbp]
	\centering
	\includegraphics[width=\linewidth]{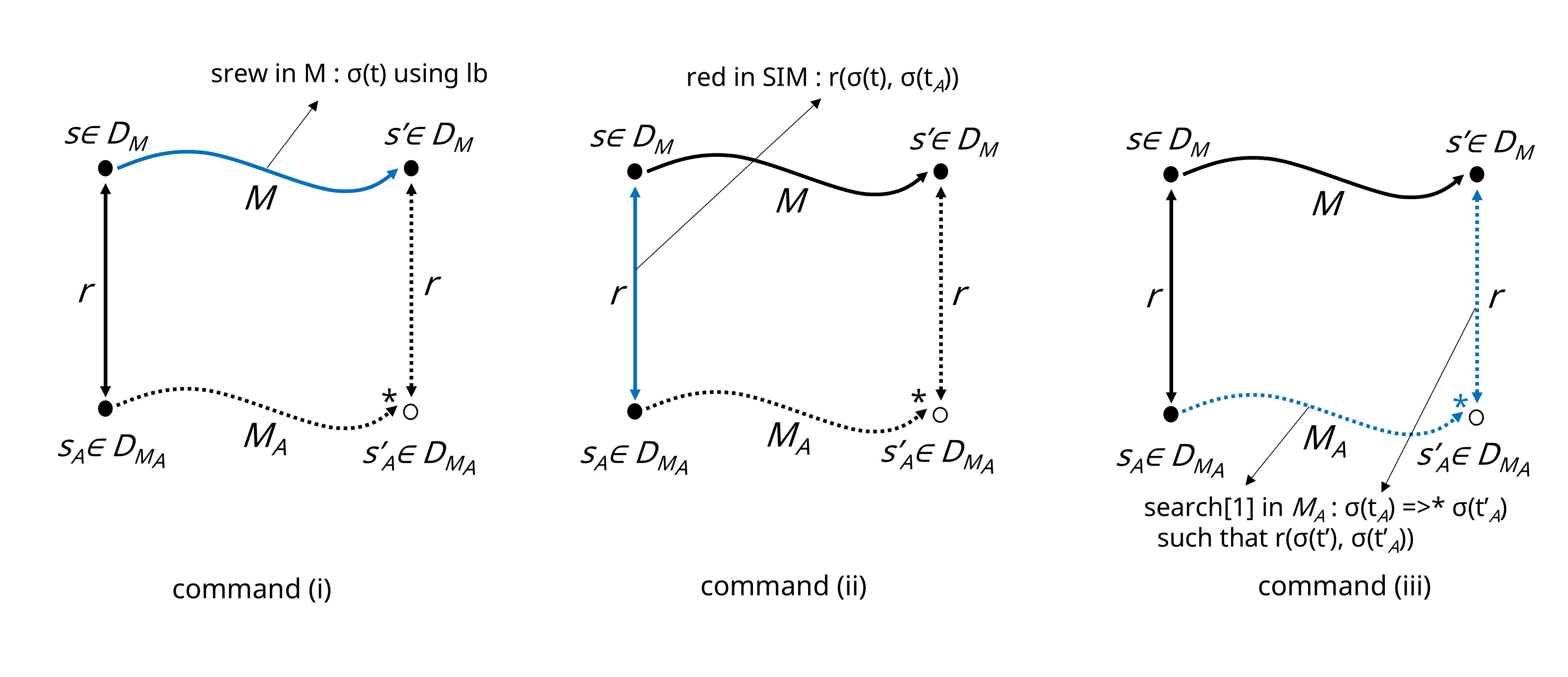}
	\caption{\textbf{Verification of simulation relation for each state transition.}}
	\label{3steps}
\end{figure}

Next, we verify that $s_A \in D_{M_A}$ corresponding to $s \in D_M$ has $r$ using the \texttt{red} command of Maude.
Specifically, human users are supposed to construct a state $s_A$ of $M_A$ for each $s$ of $M$ and use the following command:
\begin{lstlisting}
red in SIM : r(*@($\sigma(t), \sigma(t_{A})$)@*) .
\end{lstlisting}
where \texttt{SIM} is the name of the module defining the simulation relation candidate \texttt{r} with an equation.
There are separate modules defining the states of $M$ and $M_A$, both of which are imported by the module \texttt{SIM}.
This \texttt{red} command returns a Boolean value.
If the command returns true,
there is $r$ between $\sigma(t)$ and $\sigma(t_{A})$.
As with $\sigma(t)$, we use fresh constants to represent $\sigma(t_{A})$.
We construct a candidate $\sigma(t_{A})$ that is expected to hold the $r$ with $\sigma(t)$, and use the command to verify that $r$ holds between $\sigma(t)$ and $\sigma(t_{A})$.
The point at which the \texttt{reduce} command is applied is shown as command (ii) in Fig \ref{3steps}.

Finally, we verify that zero or more state transitions make it possible to move $s_A$ to $s'_A$ such that the \texttt{search} command is used to confirm that $r(s', s'_A)$ holds.
Specifically, we use the following command:
\begin{lstlisting}
search[1] in (*@$M_A$@*) : (*@$\sigma(t_{A})$@*) =>* (*@$\sigma(t_{A}')$@*) such that r(*@($\sigma(t'), \sigma(t_{A}')$)@*) .
\end{lstlisting}
where $M_A$ is the name of the module defining the state machine $M_A$.
The module $M_A$ imports the module \texttt{SIM}.
Therefore, $r$ can be used in the module $M_A$.
If this command returns a solution,
it indicates that there exists $s'_A \in D_{M_A}$ such that $r(s', s'_A)$ and $s_A \leadsto^*_{M_A} s'_A$.
The point at which the \texttt{search} command is applied is shown as command (iii) in Fig \ref{3steps}.

By the execution of our three-command procedure for each state transition caused by the rewrite rules of $M$,
we can verify that the following condition holds:
for each state transition of $M$,
for arbitrary states $s, s'$ of $M$, and an arbitrary state $s_A$ of $M_{A}$ such that $r(s, s_A)$ and $s \leadsto_{M} s'$,
there exists a state $s_A '$ of $M_{A}$ such that $r(s ', s_A ')$ and $s_A \leadsto^*_{M_A} s'_A$.
This can be expressed in a logical formula as follows:
\[
	\begin{aligned}
		 & (\forall s \in D_M)(\forall s' \in D_M)(\forall s_A \in D_{M_A}) (\exists s_A' \in D_{M_A}) \\
		 & \quad (r(s, s_A) \wedge (s, s') \in T_M)
		\Rightarrow (r(s', s'_A) \wedge s_A \leadsto^*_{M_A} s'_A)
	\end{aligned}
\]
By combining this verification result with Lemma \ref{lemma_model_checking_for_simulation},
we can formally conclude that Definition \ref{definition_simulation} (ii) is satisfied.
Specifically, since Lemma \ref{lemma_model_checking_for_simulation} ensures that the satisfaction of formula (\ref{eq:first_formula}) over the supersets is a sufficient condition for formula (\ref{eq:second_formula}),
the successful execution of our three-command procedure confirms that the state machine $M$ is simulated by $M_A$.

We have described the technique in step 3 for verifying whether both conditions of Definition \ref{definition_simulation} are satisfied.
If both conditions of Definition \ref{definition_simulation} are satisfied, then $M$ is simulated by $M_A$.

While this methodology provides a robust framework for verification using a single simulation relation, it can be further extended to address the challenges of verifying even larger and more complex systems.
In the following subsection, we introduce a methodology for the composition of simulation relations, which enables hierarchical verification by linking multiple layers of abstraction to further mitigate the state space explosion problem.

\subsection{Composition of simulation relations from state machines to state machines}\label{sec:composition}
To address cases where a single layer of abstraction is insufficient, we further extend Section~\ref{sec:combining_model_checking_with_simulation_relations} by introducing a methodology for the composition of simulation relations.
Let $M, M_A,$ and $M_B$ be state machines with increasing levels of abstraction and smaller reachable state spaces.
Let $p, p_A,$ and $p_B$ be their respective state predicates,
and let $r_1$ and $r_2$ be the simulation relation candidates from $M$ to $M_A$ and from $M_A$ to $M_B$, respectively.
Here, we assume that the reachable state space of $M_B$ is of a size that allows for model checking that $p_B$ is an invariant property wrt $M_B$,
and that both $M$ and $M_A$ have large reachable state spaces, leading to state space explosion,
making it infeasible to directly conduct model checking that $p$ is an invariant property wrt $M$ and that $p_A$ is an invariant property wrt $M_A$.
To prove that $p$ is an invariant property wrt $M$, we follow five steps:
\begin{enumerate}
	\item Conduct model checking to verify that $p_B$ is an invariant property wrt $M_B$.
	\item Conjecture a simulation relation candidate \texttt{r2} from $M_A$ to $M_B$, and prove that $p_A$ can be deduced from $p_B$ assuming the simulation relation candidate $r_2$ using static case splitting on data structures and reduction.
	\item Demonstrate that $M_A$ is simulated by $M_B$ using transition-based case splitting over rewrite rules, model checking, and reduction.
	\item Conjecture a simulation relation candidate $r$ from $M$ to $M_A$, and prove that $p$ can be deduced from $p_A$ assuming the simulation relation candidate $r_1$ using static case splitting on data structures and reduction.
	\item Demonstrate that $M$ is simulated by $M_A$ using transition-based case splitting over rewrite rules, model checking, and reduction.
\end{enumerate}
By following the steps, we can formally verify that $p$ is an invariant property wrt $M$ even though it is infeasible to directly conduct model checking that $p$ is an invariant property wrt $M$ and that $p_A$ is an invariant property wrt $M_A$.
This is useful when $M$ and $M_A$ have a large reachable state space that leads to state space explosion.
From steps 1, 2, and 3, it follows that $p_A$ is an invariant property wrt $M_A$.
This approach is consistent with the methodology described in the Section~\ref{sec:combining_model_checking_with_simulation_relations}.
From the fact that $p_A$ is an invariant property wrt $M_A$, and from steps 4 and 5, it follows that $p$ is an invariant property wrt $M$.

\section{Case study I: The Alternating Bit Protocol (ABP)}\label{sec:case_study_1}
This section demonstrates the application of our methodology to verify that the Alternating Bit Protocol (ABP) satisfies its desired invariant property.
To address \textbf{RQ1}, we establish a simulation relation from ABP to the more abstract Simple Communication Protocol (SCP) and leverage model checking on SCP to overcome the state space explosion problem.
We first provide the formal specifications of both protocols before detailing the three-step verification process.

\subsection{Formal specification of the Alternating Bit Protocol (ABP)}\label{sec:abp}
The Alternating Bit Protocol (ABP) is a communication protocol designed to provide reliable communication between a sender and a receiver, even when the underlying channels duplicate or drop messages.
The protocol architecture consists of a sender process, a receiver process, and two unreliable first-in, first-out (FIFO) channels: \texttt{queue1} for data messages and \texttt{queue2} for acknowledgments.
To maintain the communication state, the sender stores a Boolean variable \texttt{bit1} and a natural number \texttt{index}, while the receiver maintains a Boolean variable \texttt{bit2} and a natural number list \texttt{list}.
The protocol logic proceeds through the following four corresponding steps:
\begin{enumerate}
	\item The sender transmits a data message containing its current \texttt{bit1} and \texttt{index} to the receiver via \texttt{queue1}.
	\item Upon receiving a message from \texttt{queue1}, the receiver compares the message bit with its expected \texttt{bit2}. If they match, it appends the \texttt{index} to the \texttt{list}, toggles its \texttt{bit2}, and prepares for the acknowledgment; otherwise, it discards the message as a duplicate.
	\item The receiver transmits an acknowledgment message containing its current \texttt{bit2} back to the sender via \texttt{queue2}.
	\item Upon receiving an acknowledgment from \texttt{queue2}, the sender compares the acknowledgment bit with its \texttt{bit1}. If the bit confirms successful delivery, the sender toggles its \texttt{bit1} and prepares the next \texttt{index} for transmission; otherwise, it continues retransmitting the current message.
\end{enumerate}
This entire flow, built around the alternating bit, ensures reliable communication even if the underlying channels duplicate or drop messages and their acknowledgment messages.

To formalize the ABP as a state machine, we use the following observable components.
\begin{itemize}
	\item \texttt{(index: i)} - \texttt{i} is a number. The sender transmits this number to the receiver.
	\item \texttt{(list: l)} - \texttt{l} is a list of numbers. The receiver receives numbers from the sender and stores them in this list.
	\item \texttt{(queue1: q1)} - \texttt{q1} is a queue for cell1 messages, which consist of a pair containing a Boolean value and an index number. The message, encapsulated by the data type \texttt{PCell}, is dispatched by the sender to the receiver. If \texttt{queue1} is empty, \texttt{q1} is set to \texttt{nil}.
	\item \texttt{(queue2: q2)} - \texttt{q2} is a queue for cell2 messages, which contains a Boolean value. The message, encapsulated by the data type \texttt{BCell}, is transmitted from the receiver back to the sender. If \texttt{queue2} is empty, \texttt{q2} is set to \texttt{nil}.
	\item \texttt{(bit1: b1)} - \texttt{b1} is a Boolean value that the sender maintains.
	\item \texttt{(bit2: b2)} - \texttt{b2} is a Boolean value that the receiver maintains.
\end{itemize}

The initial state is defined as follows:
\begin{lstlisting}
{(index: 0) (list: nil) (queue1: nil) 
 (queue2: nil) (bit1: false) (bit2: false)} .
\end{lstlisting}

In the initial state, each value is as follows:
\begin{itemize}
	\item The \texttt{index} number is set to 0. This indicates that the sender is prepared to dispatch a cell1 message with the index number 0.
	\item The \texttt{list}, which the receiver uses to store index numbers received from the sender, is empty (\texttt{nil}). This indicates that the receiver has not yet received any index numbers.
	\item The \texttt{queue1} is empty (\texttt{nil}). This indicates that either the sender has not yet sent any cell1 messages, the receiver has already received all cell1 messages sent by the sender, or all cell1 messages have been dropped. In the initial state, this condition refers to the case where no cell1 messages have been sent.
	\item The \texttt{queue2} is empty (\texttt{nil}). This indicates that either the receiver has not yet sent any cell2 messages, the sender has received all cell2 messages from the receiver, or all cell2 messages have been dropped. In the initial state, this condition refers to the case where no cell2 messages have been sent.
	\item \texttt{bit1} is set to false, indicating the initial Boolean state for transmission control from the sender.
	\item \texttt{bit2} is also set to false, indicating the initial Boolean state for acknowledgments from the receiver.
\end{itemize}

The ABP is specified as eight rewrite rules: \texttt{send1}, \texttt{rec1}, \texttt{send2}, \texttt{rec2}, \texttt{drop1}, \texttt{dup1}, \texttt{drop2}, and \texttt{dup2}.
The rewrite rules use the following Maude variables:
\begin{itemize}
	\item \texttt{OCs} denotes a variable of observable components soups.
	\item \texttt{I} denotes a variable of \texttt{index}.
	\item \texttt{L} denotes a variable of \texttt{list}.
	\item \texttt{PQ} denotes a variable of \texttt{queue1}.
	\item \texttt{BQ} denotes a variable of \texttt{queue2}.
	\item \texttt{PC} denotes a variable of an element in \texttt{queue1}.
	\item \texttt{BC} denotes a variable of an element in \texttt{queue2}.
	\item \texttt{B}, \texttt{B1}, and \texttt{B2} denote variables of \texttt{bit1} or \texttt{bit2}.
\end{itemize}
The rewrite rule \texttt{send1} is defined as follows:
\begin{lstlisting}
crl [send1] : 
{(queue1: PQ) (bit1: B) (index: I) OCs} => 
{(queue1: (PQ pcell(B, I))) (bit1: B) (index: I) OCs}
if I <= maxIndex /\ length(PQ) <= maxQueueLength .
\end{lstlisting}
The rewrite rule specifies that if the index number, denoted as \texttt{I}, is less than or equal to the predefined \texttt{maxIndex}, and the length of \texttt{queue1}, denoted as \texttt{PQ} is less than or equal to the predefined \texttt{maxQueueLength},
the message \texttt{pcell(B, I)} is appended to the end of \texttt{queue1}.
This message encapsulates both the Boolean value (\texttt{B}) and the index number (\texttt{I}).
Both \texttt{maxIndex} and \texttt{maxQueueLength} are natural numbers.
They are used to limit the reachable state space during model checking.
This limitation is similarly applied in other rewrite rules.

The rewrite rule \texttt{rec1} is defined as follows:
\begin{lstlisting}
rl [rec1] : 
{(queue2: (BC BQ)) (bit1: B) (index: I) OCs} => 
{(queue2: BQ) 
 (bit1: (if bit(BC) = B then B else not B fi)) 
 (index: (if bit(BC) = B then I else s(I) fi)) OCs} .
\end{lstlisting}
The rewrite rule specifies that if the cell2 message queue is not empty (indicating the existence of awaiting cell2 messages within \texttt{queue2}), the sender processes the top cell2 message from this queue.
If the \texttt{bit1} maintained by the sender does not equal the bit contained in \texttt{cell2} at the front of the \texttt{queue2}, \texttt{bit1} is toggled (\texttt{not B}), and the index is incremented by 1 (\texttt{s(I)}).
Otherwise, both the \texttt{bit1} and the \texttt{index} remain unchanged.
The message at the front of \texttt{queue2} is removed.

The rewrite rule \texttt{send2} is defined as follows:
\begin{lstlisting}
crl [send2] : 
{(queue2: BQ) (bit2: B) OCs} =>
{(queue2: (BQ bcell(B))) (bit2: B) OCs}
if length(BQ) <= maxQueueLength .
\end{lstlisting}
The rewrite rule specifies that if the length of \texttt{queue2}, denoted as \texttt{BQ} is less than or equal to the predefined \texttt{maxQueueLength},
the message \texttt{bcell(B)} is appended to the end of \texttt{queue2}.
This message contains the Boolean value (\texttt{B}).

The rewrite rule \texttt{rec2} is defined as follows:
\begin{lstlisting}
rl [rec2] : 
{(queue1: (PC PQ)) (bit2: B) (list: L) OCs} => 
{(queue1: PQ) 
 (bit2: (if bit(PC) = B then not B else B fi)) 
 (list: (if bit(PC) = B then (index(PC) L) else L fi)) OCs} .
\end{lstlisting}
The rewrite rule specifies that if the cell1 message queue is not empty (indicating the existence of awaiting cell1 messages within \texttt{queue1}), the receiver processes the top cell1 message from this queue.
If the \texttt{bit2} maintained by the receiver equals the bit contained in \texttt{cell1} at the front of the \texttt{queue1}, \texttt{bit2} is toggled (\texttt{not B}), and the index is added to the list (\texttt{(index(PC) L)}).
Otherwise, both the \texttt{bit2} and the \texttt{list} remain unchanged.
The message at the front of \texttt{queue1} is removed.

The rewrite rule \texttt{drop1} is defined as follows:
\begin{lstlisting}
rl [drop1] : 
{(queue1: (PC PQ)) OCs} => {(queue1: PQ) OCs} .
\end{lstlisting}
The rewrite rule specifies that if the cell1 message queue is not empty, the message at the front of \texttt{queue1} is removed.

The rewrite rule \texttt{dup1} is defined as follows:
\begin{lstlisting}
crl [dup1] : 
{(queue1: (PC PQ)) OCs} => {(queue1: (PC PC PQ)) OCs}
if length(PQ) <= maxQueueLength .
\end{lstlisting}
The rewrite rule specifies that if the cell1 message queue is not empty, and the length of \texttt{queue1} is less than or equal to the predefined \texttt{maxQueueLength}, the message at the front of \texttt{queue1} is duplicated.

Similarly, \texttt{drop2} and \texttt{dup2} are defined for \texttt{queue2} in the same manner as \texttt{drop1} and \texttt{dup1} for \texttt{queue1}; therefore,
they are omitted to avoid redundancy.

The invariant property candidate of ABP is defined in Maude as follows:
\begin{lstlisting}
eq ABP-inv(ABP) = 
 (getBit1(ABP) = getBit2(ABP) implies 
   mk(getIndex(ABP)) = (getIndex(ABP) getList(ABP))) and 
 (getBit1(ABP) != getBit2(ABP) implies 
   mk(getIndex(ABP)) = getList(ABP)) .
\end{lstlisting}
\texttt{ABP-inv} ensures that the protocol's state remains consistent and correct.
\texttt{ABP-inv} is defined such that if \texttt{getBit1(ABP)} (the sender's bit) is equal to \texttt{getBit2(ABP)} (the receiver's bit),
then \texttt{mk(getIndex(ABP))} is equal to \texttt{(getIndex(ABP) getList(ABP))},
and if \texttt{getBit1(ABP)} is not equal to \texttt{getBit2(ABP)},
then \texttt{mk(getIndex(ABP))} is equal to \texttt{getList(ABP)}.

\subsection{Formal specification of the Simple Communication Protocol (SCP)}
The Simple Communication Protocol (SCP) is a more abstract version of ABP, where the primary distinction lies in the communication channels.
While ABP uses queues to handle multiple messages,
SCP uses two unreliable cells, \texttt{cell1} and \texttt{cell2}, which can be viewed as queues with a fixed capacity of one.
In ABP, the number of reachable states grows exponentially as parameters such as \texttt{maxQueueLength} increase, often leading to the state space explosion problem.
In contrast, by abstracting the communication channels into single-capacity cells, SCP inherently limits the complexity of the system's state configurations.
As a result, the total number of reachable states in SCP is considerably smaller than that in ABP, making direct model checking feasible and efficient without encountering state space explosion.
Due to the structural similarity between the two protocols, we define SCP by referring to the ABP specification described in the previous section as follows:
\begin{itemize}
	\item \textbf{Observable Components}: The definitions for \texttt{index}, \texttt{list}, \texttt{bit1}, and \texttt{bit2}, along with their corresponding Maude variables, are identical to those in ABP.
	      The only structural difference is that \texttt{queue1} and \texttt{queue2} in ABP are replaced by \texttt{cell1} and \texttt{cell2} in SCP, respectively.
	\item \textbf{Initial State}: The initial state of SCP follows the same logic as ABP, with both \texttt{cell1} and \texttt{cell2} initialized to \texttt{empty}.
	\item \textbf{Rewrite Rules}: SCP is specified by six rewrite rules: \texttt{send1}, \texttt{rec1}, \texttt{send2}, \texttt{rec2}, \texttt{drop1}, and \texttt{drop2}.
	      These rules are operationally analogous to the corresponding rules in ABP, but operate on a single-capacity buffer.
	\item \textbf{Omission of Duplication Rewrite Rules}: Unlike ABP, SCP does not include \texttt{dup1} or \texttt{dup2}, as the single-capacity buffer does not support message duplication in this model.
\end{itemize}

For instance, we present the rewrite rule \texttt{send2} in SCP as a representative example of this simplification:
\begin{lstlisting}
rl [send2] : 
{(cell2: BC) (bit2: B) OCs} => 
{(cell2: bcell(B)) (bit2: B) OCs} .
\end{lstlisting}
The rewrite rule specifies that the message containing \texttt{B} is put into \texttt{cell2}.
In ABP, the corresponding \texttt{send2} rule is a conditional rewrite rule (\texttt{crl}) that includes a guard condition to ensure that the current queue length does not exceed a predefined \texttt{maxQueueLength}.
In contrast, because SCP models the channel as a single-element cell, this capacity check is eliminated, and the rewrite rule becomes an unconditional rewrite rule (\texttt{rl}).
To avoid redundancy, we omit the detailed descriptions of other SCP rules, such as \texttt{send1} or \texttt{rec1}, as they follow an analogous modification pattern.

The invariant property candidate of SCP is defined in Maude as follows:
\begin{lstlisting}
eq SCP-inv(SCP) = 
 (getBit1(SCP) = getBit2(SCP) implies 
  mk(getIndex(SCP)) = (getIndex(SCP) getList(SCP))) and 
 (getBit1(SCP) != getBit2(SCP) implies 
  mk(getIndex(SCP)) = getList(SCP)) .
\end{lstlisting}
\texttt{SCP-inv} is defined using the same logical structure as \texttt{ABP-inv}.
\texttt{SCP-inv} ensures that the protocol's state remains consistent and correct.
Similarly to \texttt{ABP-inv}, \texttt{SCP-inv} is defined such that
if \texttt{getBit1(SCP)} is equal to \texttt{getBit2(SCP)}, then \texttt{mk(getIndex(SCP))} is equal to \texttt{(getIndex(SCP) getList(SCP))},
and if \texttt{getBit1(SCP)} is not equal to \texttt{getBit2(SCP)}, then \texttt{mk(getIndex(SCP))} is equal to \texttt{getList(SCP)}.

\subsection{Verification Results}\label{sec:case_study_1_results}
In this section, we demonstrate that ABP satisfies its invariant property,
\texttt{ABP-inv}, by leveraging the abstract model SCP and the simulation relation established between the two protocols.
Following the three-step methodology illustrated in Fig \ref{overview}, we first conduct model checking on the abstract state machine, SCP (Step 1).

\subsubsection{Step 1: Verification of Invariant Property on SCP}
In this verification, we set the parameter to \texttt{maxIndex}$\ = 64$.
We can check whether SCP satisfies the invariant property candidate, \texttt{SCP-inv}, or not by using the following \texttt{search} command:
\begin{lstlisting}
search[1] in SCP : init =>* 
 {(bit1: B1) (bit2: B2) (index: I) (list: L) OCs} 
such that not SCP-inv(scp(B1, I, B2, L)) .
\end{lstlisting}
Maude returns \texttt{No solution}.
The result \texttt{No solution} formally verifies that \texttt{SCP-inv} holds throughout all reachable states of the protocol.

\subsubsection{Step 2: Deducibility of ABP-inv from SCP-inv}
We conjecture a simulation relation candidate \texttt{r2} from ABP to SCP.
This simulation relation candidate requires that if \texttt{SCP-inv} is an invariant property wrt SCP and ABP is simulated by SCP, then \texttt{ABP-inv} is an invariant property wrt ABP.
The simulation relation candidate from ABP to SCP is defined as follows:
\begin{lstlisting}
eq r2(ABP, SCP) = 
 (getBit1(ABP) = getBit1(SCP)) and 
 (getBit2(ABP) = getBit2(SCP)) and 
 (getIndex(ABP) = getIndex(SCP)) and 
 (getList(ABP) = getList(SCP)) .
\end{lstlisting}
This demonstrates that \texttt{bit1}, \texttt{bit2}, \texttt{index}, and \texttt{list} are equivalent in ABP and SCP.
The parameter \texttt{ABP} is defined as a value of the data type \texttt{ABPState}, which represents the state of the Alternating Bit Protocol.
Similarly, the parameter \texttt{SCP} is defined as a value of the data type \texttt{SCPState}, representing the state of the Simple Communication Protocol.
Note that \texttt{r2} omits the contents of the communication channels, as they are not required to deduce \texttt{ABP-inv} from \texttt{SCP-inv}.
This mapping is sufficient because both invariant properties depend only on the sender and receiver's local states.

We demonstrate that if \texttt{SCP-inv} is an invariant property wrt SCP and ABP is simulated by SCP, then \texttt{ABP-inv} is an invariant property wrt ABP.
We confirm this using the following function.
\begin{lstlisting}
eq ABP2SCP-sim(ABP, SCP) = 
 r2(ABP, SCP) implies (SCP-inv(SCP) implies ABP-inv(ABP)) .
\end{lstlisting}

We describe the proof of \texttt{SCP-inv}($s_{SCP}$) $\Rightarrow$ \texttt{ABP-inv}($s_{ABP}$) for arbitrary states $s_{SCP}$, $s_{ABP}$ with \texttt{r2}($s_{ABP}$, $s_{SCP}$).
Considering the simulation relation candidate and the invariant property, the cases can be divided as follows.
\begin{enumerate}
	\item \texttt{r2} is true, \texttt{SCP-inv} is true, and \texttt{bit1} is equal to \texttt{bit2}.
	\item \texttt{r2} is true, \texttt{SCP-inv} is true, and \texttt{bit1} is different from \texttt{bit2}.
	\item \texttt{r2} is true, \texttt{SCP-inv} is false, and \texttt{bit1} is equal to \texttt{bit2}.
	\item \texttt{r2} is true, \texttt{SCP-inv} is false, and \texttt{bit1} is different from \texttt{bit2}.
	\item \texttt{r2} is false.
\end{enumerate}
The module and the corresponding \texttt{red} command for case 2 are provided below:
\begin{lstlisting}
fmod SIM_ABP2SCP_02 is
  pr SIM .
  
  ops b1 b2 b3 b4 : -> Bool+ .
  ops i1 i2 : -> Nat+ .
  ops l1 l2 : -> List{Nat+} .
  
  eq (b1 = b2) = false .
  eq (b1 = b3) = true .
  eq (b2 = b4) = true .
  eq (b3 = b4) = false .
  eq (i1 =Nat+ i2) = true .
  eq l1 = mk(i1) .
  eq l2 = mk(i2) .
  eq (mk(i1) = mk(i2)) = true .
endfm
\end{lstlisting}

As shown in the module \texttt{SIM\_ABP2SCP\_02}, we introduce several fresh constants.
In this context, \texttt{Nat+} and \texttt{Bool+} are data types used to represent arbitrary natural numbers and Boolean values, respectively,
enabling the symbolic verification of each case.
\begin{lstlisting}
red in SIM_ABP2SCP_02 : 
 ABP2SCP-sim(abp(pqueue, bqueue, b1, i1, b2, l1), 
             scp(pcell, bcell, b3, i2, b4, l2)) .
\end{lstlisting}
The constants \texttt{b1}, \texttt{b2}, \texttt{b3}, and \texttt{b4} represent arbitrary Boolean values,
and \texttt{i1} and \texttt{i2} denote arbitrary natural numbers.
Similarly, \texttt{l1} and \texttt{l2} represent arbitrary lists of natural numbers.
For the communication channels in ABP, \texttt{pqueue} and \texttt{bqueue} represent arbitrary values for the message queues from the sender to the receiver and from the receiver to the sender, respectively.
For the corresponding cells in SCP, \texttt{pcell} and \texttt{bcell} represent arbitrary values from the sender to the receiver and from the receiver to the sender, respectively.
This \texttt{red} command returns true.
Although case 2 is presented here as a representative demonstration,
all other cases listed above were similarly verified using \texttt{red} command.
All cases yielded the result true,
thereby formally establishing the deducibility of \texttt{ABP-inv} from \texttt{SCP-inv} under the simulation relation \texttt{r2}.

\subsubsection{Step 3: Demonstration of Simulation Relation from ABP to SCP}
First, we verify Definition \ref{definition_simulation} (i) is satisfied.
we confirm that the simulation relation holds between the initial states of ABP and SCP as follows:
\begin{lstlisting}
red in SIM : 
 r2(abp(nil, nil, false, 0, false, nil), 
    scp(empty, empty, false, 0, false, nil)) .
\end{lstlisting}
Maude returns true for this \texttt{red} command.
This means that the simulation relation \texttt{r2} holds between the initial states of ABP and SCP.

To verify that Definition \ref{definition_simulation} (ii) is satisfied,
we apply the verification procedure consisting of Maude commands (i), (ii), and (iii) as depicted in Fig \ref{3steps} to each rewrite rule of ABP.
This process involves case splitting for each of the eight rewrite rules to ensure that all potential transitions are formally accounted for.
The full set of these generated cases and the corresponding Maude code are publicly available in our repository.
For illustrative purposes, we present the verification of the \texttt{send2} rule as a representative example.

Before executing the commands, we define the representative concrete state ($s_{ABP}$) and its corresponding abstract state ($s_{SCP}$) using fresh constants to ensure the generality of the verification.
Let $s_{ABP}$ be the ABP state for this case:
\begin{lstlisting}
{(queue1: pqueue) (queue2: bqueue) (bit1: bool) 
 (index: index) (bit2: bool) (list: list)}
\end{lstlisting}
The corresponding SCP state, $s_{SCP}$, is constructed as follows:
\begin{lstlisting}
{(cell1: pcell) (cell2: bcell) (bit1: bool) 
 (index: index) (bit2: bool) (list: list)}
\end{lstlisting}
In these terms, \texttt{pqueue}, \texttt{bqueue}, \texttt{pcell}, \texttt{bcell}, \texttt{bool}, \texttt{index}, and \texttt{list} are fresh constants representing arbitrary values.
While \texttt{index} is constrained to be less than or equal to \texttt{maxIndex},
the remaining constants represent unconstrained and arbitrary values.

The procedure begins with command (i), where the \texttt{srew} command is applied to $s_{ABP}$ to identify the successor state in ABP.
Maude identifies a single successor state $s'_{ABP}$ resulting from the application of the \texttt{send2} rule:
\begin{lstlisting}
{(queue1: pqueue) (queue2: (bqueue bcell(bool))) (bit1: bool)
 (index: index) (bit2: bool) (list: list)}
\end{lstlisting}

Next, in command (ii), we verify that the simulation relation \texttt{r2} holds between the states $s_{ABP}$ and $s_{SCP}$ by executing a \texttt{red} command.
Maude evaluates this command to true, confirming that the simulation relation \texttt{r2} holds between $s_{ABP}$ and $s_{SCP}$.

Finally, we execute command (iii) to search for a state $s'_{SCP}$ in SCP that preserves the simulation relation with $s'_{ABP}$ in zero or more steps.
Specifically, we use a \texttt{search} command to find a state transition from $s_{SCP}$ to a state that satisfies \texttt{r2}($s'_{ABP} ,\ s'_{SCP}$).
Maude successfully finds a solution for this search, thereby identifying a valid state transition in the abstract model that corresponds to the concrete state transition.

These collective results demonstrate that,
for all state transitions induced by the \texttt{send2} rule in ABP,
there exist corresponding transitions in SCP that preserve the simulation relation.
By exhaustively applying this three-command procedure to the remaining seven rewrite rules,
we formally confirm that the state transition condition in Definition \ref{definition_simulation} (ii) is fully satisfied for ABP.

\subsubsection{Summary of Verification}
Through the successful execution of Steps 1, 2, and 3,
we have formally demonstrated that the invariant property \texttt{SCP-inv} holds for SCP,
that \texttt{ABP-inv} is logically deducible from \texttt{SCP-inv} under the simulation relation \texttt{r2},
and that ABP is simulated by SCP.
Based on Theorem \ref{theorem1}, these results collectively prove that the invariant property \texttt{ABP-inv} holds for ABP.

\subsection{Discussion}
We conducted the experiments described in this paper on a computer equipped with a 3.3 GHz microprocessor and 32 GB of memory.
Maude 3.3.1 was used for verification across all case studies.

Table~\ref{performance_comparison} summarizes the performance comparison between direct model checking and our proposed methodology.
To evaluate the effectiveness of our approach (\textbf{RQ1}), we first attempted direct model checking on the ABP state machine.
With parameters set to \texttt{maxIndex}$\ = 64$ and \texttt{maxQueueLength}$\ = 32$, the process failed to conclude.
Preliminary tests indicated that the verification time increases sharply as these parameters grow.
The chosen values represent a configuration where state space explosion renders exhaustive direct model checking infeasible within a reasonable timeframe, effectively highlighting the practical boundaries of conventional methods.
We deliberately terminated the verification after one hour, judging that state space explosion rendered direct model checking infeasible within a practical timeframe.

\begin{table}[!ht]
	\begin{adjustwidth}{-0.5in}{0in}
		\centering
		\caption{
			\textbf{Performance comparison of verification methods for ABP (maxIndex $=$ 64, maxQueueLength $=$ 32).}
		}
		\label{performance_comparison}
		\begin{tabular}{|l|l|l|l|l|}
			\hline
			\textbf{Verification Method} & \textbf{Step} & \textbf{Maude Command}                       & \textbf{Time} & \textbf{Result}           \\ \thickhline
			Direct Model Checking        & -             & \texttt{search}                              & $>$ 3,600 s   & \textbf{Timeout (Failed)} \\ \hline
			Proposed Methodology         & Step 1        & \texttt{search} (on SCP)                     & $<$ 1 s       & Verified                  \\ \cline{2-5}
			                             & Step 2        & \texttt{red}                                 & $<$ 1 s       & Verified                  \\ \cline{2-5}
			                             & Step 3        & \texttt{srew}, \texttt{red}, \texttt{search} & $<$ 1 s       & Verified                  \\ \hline
			% \textbf{Total (Proposed)}    & -             & -                                            & < 3 s         & \textbf{Verified}         \\ \hline
		\end{tabular}
		\label{table1}
	\end{adjustwidth}
\end{table}

In contrast, our proposed methodology completed the entire verification process with minimal computational overhead.
Specifically, Steps 1, 2, and 3 each finished in less than one second.
The high efficiency of Step 1 with \texttt{maxIndex}$\ = 64$ is attributed to the significantly reduced state space of SCP, which abstracts queues into single-capacity cells. Furthermore,
since Steps 2 and 3 rely on symbolic execution with fresh constants rather than exhaustive state space traversal, their execution times remain sub-second and independent of the parameter sizes that typically cause state space explosion.
These results provide a definitive affirmative answer to \textbf{RQ1}, proving that our methodology can mitigate the state space explosion problem in ABP effectively.

While our methodology effectively mitigates the state space explosion problem, it introduces the requirement for manual abstraction and the definition of simulation relations.
To minimize this burden, we developed a support tool \cite{DBLP:conf/sfpvv/IshibashiNO25}.
By automating the case splitting arising from associative commutative pattern matching and the execution of the three-command sequence (\texttt{srew}, \texttt{red}, and \texttt{search}) in step 3, the tool reduces the risk of human error and ensures that the verification process remains rigorous.
The tool is publicly available \cite{DBLP:conf/sfpvv/IshibashiNO25}.

\section{Case study II: The Sliding Window Protocol (SWP)}\label{sec:case_study_2}
This section demonstrates the formal verification of the Sliding Window Protocol (SWP) to address \textbf{RQ2}.
To overcome this, we employ the composition of simulation relations as established in Section~\ref{sec:composition}.
Specifically, we prove that SWP satisfies its invariant property by establishing a simulation relation from SWP to ABP and composing it with the previously verified simulation from ABP to SCP .
This hierarchical approach allows us to deduce the correctness of SWP from the abstract SCP invariant property through the intermediate ABP model.

\subsection{Formal specification of the Sliding Window Protocol (SWP)}
The Sliding Window Protocol (SWP) is a fundamental method
for reliable and sequential data transmission over a communication channel.
This protocol works by allowing the sender to transmit multiple messages before receiving an acknowledgement from the receiver, effectively managing the flow of data.
This protocol has a window mechanism where the window size is the number of messages that can be sent without an acknowledgement.
There are differences between SWP and ABP.
In SWP, the sender can send messages up to the window size without waiting for an acknowledgement from the receiver.
On the other hand, in ABP, the sender needs to wait for an acknowledgement from the receiver before sending each message.

To maintain conciseness and clarity,
we define the formal specification of SWP by highlighting its extensions and differences compared to the ABP model described in Section~\ref{sec:abp}.
\begin{itemize}
	\item \textbf{Observable Components}: While the basic communication structure involving \texttt{index}, \texttt{list}, and unreliable channels (\texttt{queue1}, \texttt{queue2}) remain the same,
	      SWP introduces several components to manage its window mechanism.
	      These include \texttt{seq} (the current sequence number), \texttt{windowCount} (the number of unacknowledged frames),
	      \texttt{expectedFrameSeq} (the next sequence number the receiver expects), and \texttt{expectedAckSeq} (the next acknowledgment the sender expects).
	\item \textbf{Maude Variables}: In the formalization of SWP, the variables \texttt{OCs} (observable components), \texttt{I} (index), and \texttt{L} (list) are identical to those used in the ABP specification.
	      However, other variables are uniquely defined for SWP to reflect its specific data structures,
	      such as \texttt{FQ} and \texttt{AQ} for the queues,
	      and \texttt{S}, \texttt{C}, \texttt{F}, and \texttt{A} for managing sequence numbers and window counts.
	\item \textbf{Initial State}: The initial state of SWP follows a similar logic to ABP, with all sequence-related components initialized to zero and communication channels set to empty.
	\item \textbf{Rewrite Rules}: SWP is specified by eight rewrite rules.
	      The rules \texttt{send1}, \texttt{rec1}, \texttt{send2}, and \texttt{rec2} are modified to incorporate sequence numbers and window management logic.
	      Notably, the rewrite rules \texttt{drop1}, \texttt{drop2}, \texttt{dup1}, and \texttt{dup2} for SWP are operationally identical to those defined for ABP in Section \ref{sec:abp}.
	      Therefore, their detailed specifications are omitted here to avoid redundancy.
\end{itemize}

The rewrite rule \texttt{send1} is defined as follows:
\begin{lstlisting}
crl [send1] : 
{(queue1: FQ) (seq: S) (index: I) (windowCount: C) OCs} => 
{(queue1: (FQ frame(S, I))) (seq: next(S)) 
 (index: s(I)) (windowCount: s(C)) OCs}
if C < windowSize /\ I <= maxIndex /\ 
   length(FQ) <= maxQueueLength .
\end{lstlisting}
The rewrite rule specifies that if the window count, denoted as \texttt{C}, is less than the predefined \texttt{windowSize}, the index number, denoted as \texttt{I}, is less than or equal to the predefined \texttt{maxIndex}, and the length of \texttt{FQ} is less than or equal to the predefined \texttt{maxQueueLength},
the message \texttt{frame(S, I)} is appended to the end of \texttt{queue1}.
This message encapsulates both the sequence number (\texttt{S}) and the index number (\texttt{I}).
The sender increments the index number and the window count by 1, and updates the sequence number.
In this formal specification, \texttt{windowSize} is set to 16.

The \texttt{next} function is defined as follows:
\begin{lstlisting}
eq next(N) = (N + 1) rem windowSize .
\end{lstlisting}
This function calculates the next number for a given number \texttt{N}. Specifically, it adds 1 to \texttt{N} and then takes the remainder of dividing this sum by \texttt{windowSize}.
This operation is used to cycle numbers within a specific range defined by \texttt{windowSize}.
This mechanism implies that the sequence numbers follow a circular order rather than a linear one.
For instance, if \texttt{windowSize} is 16, the numbers cycle from 0 to 15 and back to 0, where 0 is considered the successor of 15 in this circular sequence.

The rewrite rule \texttt{rec1} is defined as follows:
\begin{lstlisting}
rl [rec1] : 
{(queue2: (AC AQ)) (expectedAckSeq: A) 
 (seq: S) (windowCount: C) OCs} => 
{(queue2: AQ) (expectedAckSeq: expectedAckSeq(A, seq(AC), S)) 
 (seq: S) (windowCount: windowCount(C, A, seq(AC), S)) OCs} .
\end{lstlisting}
The rewrite rule specifies that if the ack message queue is not empty (indicating the existence of awaiting ack messages within \texttt{queue2}),
the sender processes the top ack message from this queue.
The sender updates the \texttt{expectedAckSeq} to reflect the sequence number of the next expected ack message,
and updates the \texttt{windowCount} accordingly to manage the flow of outgoing messages.
The message at the front of \texttt{queue2} is removed.

To describe the \texttt{expectedAckSeq} function and the \texttt{windowCount} function, we first discuss the \texttt{between} function.
This is because both the \texttt{expectedAckSeq} function and the \texttt{windowCount} function use the \texttt{between} function.
The \texttt{between} function is defined as follows:
\begin{lstlisting}
eq between(N1, N2, N3) = 
 ((N1 <= N2) and (N2 < N3)) or 
 ((N3 < N1) and (N1 <= N2)) or 
 ((N2 < N3) and (N3 < N1)) .
\end{lstlisting}
A circular order governs the relationship between these numbers, where the interval from \texttt{N1} to \texttt{N3} may wrap around the maximum sequence value.
The \texttt{between} function evaluates whether a given sequence number falls within this circular arc, which is essential for deciding whether a received frame message should be accepted or discarded.
It returns true under any of the following conditions:
\begin{itemize}
	\item \texttt{N1} is less than or equal to \texttt{N2}, and \texttt{N2} is less than \texttt{N3}.
	\item \texttt{N3} is less than \texttt{N1}, and \texttt{N1} is less than or equal to \texttt{N2}.
	\item \texttt{N2} is less than \texttt{N3}, and \texttt{N3} is less than \texttt{N1}.
\end{itemize}

\texttt{expectedAckSeq} function is defined as follows:
\begin{lstlisting}
ceq expectedAckSeq(N1, N2, N3) = 
 expectedAckSeq(next(N1), N2, N3) if between(N1, N2, N3) .
ceq expectedAckSeq(N1, N2, N3) = N1 if not between(N1, N2, N3) .
\end{lstlisting}
The \texttt{expectedAckSeq} function calculates the expected ack sequence number given three numbers \texttt{N1}, \texttt{N2}, and \texttt{N3}. It uses the between function to determine if \texttt{N2} lies between \texttt{N1} and \texttt{N3}. The process is as follows:
\begin{itemize}
	\item If \texttt{N2} is between \texttt{N1} and \texttt{N3} (\texttt{between(N1, N2, N3)} is true), the function recursively calls itself with the next number of \texttt{N1} (\texttt{next(N1)}), \texttt{N2}, and \texttt{N3} as the new arguments. This recursion continues until \texttt{N2} is no longer between \texttt{N1} and \texttt{N3}.
	\item If \texttt{N2} is not between \texttt{N1} and \texttt{N3}, it simply returns \texttt{N1}.
\end{itemize}

The \texttt{windowCount} function is defined as follows:
\begin{lstlisting}
ceq windowCount(N1, N2, N3, N4) = 
 windowCount(sd(N1, 1), next(N2), N3, N4) if between(N2, N3, N4) .
ceq windowCount(N1, N2, N3, N4) = N1 if not between(N2, N3, N4) .
\end{lstlisting}
The \texttt{windowCount} function, which involves four numbers \texttt{N1}, \texttt{N2}, \texttt{N3}, and \texttt{N4}, updates \texttt{N1} based on specific conditions. This function is used for determining the window count within a certain numerical range, using the between function to check if \texttt{N3} lies between \texttt{N2} and \texttt{N4}. The process is as follows:
\begin{itemize}
	\item If \texttt{N3} is between \texttt{N2} and \texttt{N4} (\texttt{between(N2, N3, N4)} is true), the function recursively calls itself with \texttt{N1} updated to \texttt{sd(N1, 1)} (subtracting 1 from \texttt{N1}), \texttt{N2} to its next number \texttt{next(N2)}, and keeps \texttt{N3} and \texttt{N4} as they are. This recursion continues until the condition becomes false.
	\item If \texttt{N3} is not between \texttt{N2} and \texttt{N4}, it simply returns \texttt{N1}.
\end{itemize}

The rewrite rule \texttt{send2} is defined as follows:
\begin{lstlisting}
crl [send2] : 
{(queue2: AQ) (expectedFrameSeq: S) (list: L) OCs} => 
{(queue2: (AQ ack(prev(S)))) (expectedFrameSeq: S) (list: L) OCs}
if length(L) > 0 /\ length(AQ) <= maxQueueLength .
\end{lstlisting}
The rewrite rule specifies that if the list, denoted as \texttt{L}, is not empty, and the length of \texttt{AQ} is less than or equal to the predefined \texttt{maxQueueLength},
the message \texttt{ack(prev(S))} is appended to the end of \texttt{queue2}.
This message encapsulates the previous sequence number.

The \texttt{prev} function is defined as follows:
\begin{lstlisting}
eq prev(N) = (N + sd(windowSize, 1)) rem windowSize . 
\end{lstlisting}
The function \texttt{prev} calculates the preceding number of a given number \texttt{N} within a cyclic range defined by \texttt{windowSize}. The \texttt{sd(windowSize, 1)} subtracts 1 from \texttt{windowSize}.
The function adds the result of \texttt{sd(windowSize, 1)} to \texttt{N}, and then takes the remainder of dividing this sum by \texttt{windowSize}. This operation ensures that the calculation cycles through the numbers within the range of windowSize.
For instance, if \texttt{windowSize} is 16, this function cycles numbers between 0 and 15.
If \texttt{N} is 0, the previous number is 15.

The rewrite rule \texttt{rec2} is defined as follows:
\begin{lstlisting}
rl [rec2] : 
{(queue1: (FR FQ)) (expectedFrameSeq: S) (list: L) OCs} => 
{(queue1: FQ) 
 (expectedFrameSeq: (if seq(FR) = S then next(S) else S fi)) 
 (list: (if seq(FR) = S then (index(FR) L) else L fi)) OCs} .
\end{lstlisting}
The rewrite rule specifies that if the frame message queue is not empty (indicating the existence of awaiting frame messages within \texttt{queue1}), the receiver processes a frame message from this queue.
If the sequence number contained in the frame message aligns with the expected number, the receiver then updates the \texttt{expectedFrameSeq} to reflect the sequence number of the next expected frame message,
and stores the index number to the list.
The message at the front of \texttt{queue1} is removed.

In this paper, the invariant property candidate of SWP is defined as the receiver storing the index numbers sent by the sender in the list in order, without duplication or loss.
The invariant property candidate of SWP is defined in Maude as follows:
\begin{lstlisting}
eq SWP-inv(SWP) = subsequence(mk(getIndex(SWP)), getList(SWP)) .
\end{lstlisting}
\texttt{SWP-inv} is defined such that the list of messages stored by the receiver is a subsequence of the list of messages sent by the sender.
The \texttt{subsequence} function is defined to return true if the second list can be derived from the first by deleting zero or more elements while strictly preserving their relative order.
This definition mathematically ensures that the messages are neither duplicated, lost, nor reordered.

\subsection{Verification Results}
We conjecture a simulation relation candidate from the state machine SWP to the state machine ABP.
This simulation relation requires that if \texttt{ABP-inv} is an invariant property wrt ABP and SWP is simulated by ABP, then \texttt{SWP-inv} is an invariant property wrt SWP.

The simulation relation candidate from SWP to ABP is defined as follows:
\begin{lstlisting}
eq r1(SWP, ABP) = 
 (getIndex(SWP) >= getIndex(ABP)) and 
 (getList(SWP) = getList(ABP)) .
\end{lstlisting}
This shows that the index of SWP is greater than or equal to that of ABP, and the list of SWP is equivalent to the list of ABP.
This inequality captures the fundamental behavioral difference between the two protocols: while the sender in ABP must wait for an acknowledgment before incrementing the index for the next message, the sender in SWP can transmit multiple messages ahead up to the window size.
Consequently, the index in SWP can advance beyond that of the more abstract ABP, and this relation accurately abstracts such advanced states.
The argument \texttt{SWP} is a value of the data type \texttt{SWPState} representing the state of the Sliding Window Protocol.
The argument \texttt{ABP} is a value of the data type \texttt{ABPState} representing the state of the Alternating Bit Protocol.
While SWP includes various internal components for window management,
\texttt{r1} omits them as they are not needed to deduce \texttt{SWP-inv} from \texttt{ABP-inv}.
This focused definition is sufficient for establishing a formal bridge between the two models.

We demonstrate that if \texttt{ABP-inv} is an invariant property wrt ABP and SWP is simulated by ABP, then \texttt{SWP-inv} is an invariant property wrt SWP.
We confirm this using the following function.
\begin{lstlisting}
eq SWP2ABP-sim(SWP, ABP) = 
 r1(SWP, ABP) implies (ABP-inv(ABP) implies SWP-inv(SWP)) .
\end{lstlisting}
We verified this implication by employing case splitting and the \texttt{red} command with fresh constants, similar to the process described in Case study I.
All cases yielded the result true, thereby confirming the deducibility of the SWP invariant.

Having established this deducibility, we then proceed to Step 3: demonstrating that SWP is simulated by ABP.
By applying the three-command verification procedure established in Section~\ref{sec:combining_model_checking_with_simulation_relations}---comprising (i) \texttt{srew}, (ii) \texttt{red}, and (iii) \texttt{search} commands---in the same manner as in Case study I,
we formally demonstrate that the simulation relation \texttt{r1} is preserved across all state transitions.

This result confirms that for every state transition in SWP,
there exists a corresponding sequence of transitions in ABP that preserves \texttt{r1}.
Since we have already established in Case study I that \texttt{ABP-inv} is an invariant property of ABP,
we can conclude that \texttt{SWP-inv} is an invariant property of SWP.
This conclusion is formally supported by the composition of simulation relations \texttt{r1} and \texttt{r2} as per Theorem \ref{theorem1}.

\subsection{Discussion}
Our experimental results provide a clear affirmative answer to \textbf{RQ2}.
For the verification experiments, we configured the parameters for SWP as \texttt{maxIndex}$\ = 64$, \texttt{maxQueueLength}$\ = 32$, and \texttt{windowSize}$\ = 16$, whereas those for ABP were set to \texttt{maxIndex}$\ = 64$ and \texttt{maxQueueLength}$\ = 32$.
Under these configuration, the reachable state space of SWP is significantly larger than that of the ABP model used in Case study I, rendering direct model checking entirely infeasible due to the state space explosion problem.
However, by decomposing the verification into two layers of simulation (from SWP to ABP and from ABP to SCP),
we successfully mitigated the state space explosion that makes direct model checking of SWP infeasible.

In Case study I, the invariant properties \texttt{ABP-inv} and \texttt{SCP-inv} were nearly identical in their logical structure and conceptual formulation.
In contrast, Case study II exhibited a significant disparity between the invariant properties of the concrete and intermediate models.
While \texttt{SWP-inv} is defined as a subsequence requirement---ensuring that the received messages exist within the sequence of sent messages---\texttt{ABP-inv} is formulated based on the specific state-matching logic of the alternating bit.
Despite this clear difference between the two properties,
we successfully verified \texttt{SWP-inv} by establishing a simulation relation that effectively bridges these distinct conceptual layers.
This result highlights the versatility of our proposed framework: it is not limited to cases where the abstract and concrete properties are similar,
but can also formally prove invariants even when the gap between the properties of the abstraction layers is significant.

Our approach offers practical advantages by leveraging domain knowledge to tailor abstractions, providing greater flexibility.
Furthermore, since many distributed protocols share fundamental requirements such as data consistency and ordering guarantees, our technique facilitates the reuse of previously verified formal specifications as abstract models for new targets.
This reusability, demonstrated by our hierarchical verification of SWP using the ABP and SCP models, provides a scalable and efficient pathway for the formal verification of complex systems.

\section{Related work}\label{sec:related_work}
``Refinement'' serves as a top-down development paradigm that enables the systematic derivation of concrete implementations from abstract specifications through the stepwise augmentation of architectural details.
This paradigm traces its origins to the foundational work of Wirth \cite{DBLP:journals/cacm/Wirth83}, Hoare \cite{DBLP:conf/ac/Hoare75a}, Dijkstra \cite{DBLP:journals/cacm/Dijkstra75}, and Abadi et al. \cite{DBLP:journals/tcs/AbadiL91}.
Building upon these theoretical underpinnings, several formal methods and frameworks have been established, including Vienna Development Method (VDM) \cite{DBLP:books/daglib/0068091}, Z notation \cite{DBLP:books/daglib/0072139}, Abstract State Machines (ASM) \cite{DBLP:journals/fac/Borger03}, B-Method \cite{DBLP:books/daglib/0015096}, and Event-B \cite{DBLP:books/daglib/0024570}.
Recently, there has been growing interest in using Large Language Models (LLMs) for refinement, as exemplified by the work of Cai et al. \cite{DBLP:journals/pacmpl/CaiHSLLSD25} on LLM-guided code generation.

``Simulation'', in the algebraic sense defined by Milner \cite{DBLP:conf/ijcai/Milner71}, formalizes the principle of abstracting away implementation-specific details, such as data representation and control sequencing, to identify the core behavior shared by distinct programs.
This concept is characterized by a simulation relation acting as a quasi-ordering between programs; mutual simulation establishes an equivalence relation that facilitates abstraction from implementation-specific details.
Simulation is extended to concurrent computing through process calculi such as CSP \cite{hoare} and CCS \cite{milner1, milner2}.
Within these formalisms, bisimulation---a mutual simulation where processes mimic each other's transitions---became a cornerstone for formalizing behavioral equivalence.
This ensures that disparate implementations maintain identical observable behaviors in concurrent environments.
Subsequently, simulation is applied to the analysis of I/O automata to verify that concrete systems satisfy the trace properties of their abstract specifications \cite{DBLP:conf/podc/LynchT87, lynch}

Ogata and Futatsugi\cite{ogata2008simulation} have proposed a methodology for verifying that an observational transition system (OTS) $S$ maintains invariant properties through a simulation from $S$ to another OTS, which is more abstract than $S$, in the OTS/CafeOBJ method\cite{ogata2003proof}.
They concluded that there is no definitive evidence to suggest the superiority of one method over the other.
Subsequently, Tran et al. \cite{duong} demonstrated the efficacy of simulation-based invariant verification through case studies on the MCS and Anderson mutual exclusion protocols.
Using CafeOBJ \cite{diaconescu1998cafeobj} for formal modeling and proof, they reported that the simulation-based technique enabled a more efficient verification process for the MCS protocol than the induction-based alternative.

The validity of refinement or simulation relations is predominantly established through theorem proving.
Whether performed manually or with the assistance of interactive provers, these deductive approaches require significant human expertise and effort to construct rigorous proofs.
In contrast, the approach presented in this paper deviates from this convention by using model checking to verify that a simulation holds between state machines.

There have been several papers related to the verification of the Sliding Window Protocol (SWP).
Early work by Jonsson \cite{DBLP:conf/podc/Jonsson87} established a foundation for SWP verification using simulation, though the proofs remained manual and did not use model checker or theorem prover.
Similarly, Paliwoda and Sanders \cite{DBLP:journals/dc/PaliwodaS91} employed CSP to specify SWP incrementally by composing simpler components like the Alternating Bit Protocol (ABP).
Their approach remained a manual effort without the use of model checker or theorem prover.
Later efforts integrated model checking to improve automation.
Kaivola \cite{DBLP:conf/cav/Kaivola97} used compositional preorders to verify SWP, effectively reducing the domain of values for transmitted data items to a small set (e.g., $\{0, 1, 2\}$) to facilitate model checking.
Nevertheless, this approach still required manual proofs to generalize the findings to arbitrary sequence lengths.
Direct model checking of SWP has also been explored using NuSMV; for instance, Zhao et al. \cite{DBLP:journals/jcp/ZhaoYXL09} successfully verified the protocol but were forced to impose strict parameter constraints, such as a window size of 4, channel lengths of 8 for the data queue from sender to receiver and 16 for the acknowledgment queue from receiver to sender, and boolean-only messages.
In contrast, Case study II of this paper demonstrates that our proposed methodology can handle substantially larger parameters---including a window size of 16, channel capacities of 32, and message values up to 64 ---by combining the strengths of simulation and automated model checking in Maude.
This hierarchical approach eliminates the need for manual inductive proofs for generalization or strict algebraic constraints, providing a more scalable and automated pathway for verifying complex, practical protocols.

State space explosion is one of the most challenging problems in model checking.
Numerous researchers have addressed this problem and developed a variety of techniques to mitigate it.
Clarke et al. summarize several techniques that address the state space explosion problem in model checking \cite{Clarke2012}.
They introduce techniques such as symbolic model checking with Binary Decision Diagrams (BDDs), partial order reduction, counterexample-guided abstraction refinement (CEGAR), and bounded model checking.

Abstraction, in the context of model checking, aims to reduce the system's state space by omitting the irrelevant details \cite{edmund1, edmund2, jose1, yin1, yin2, bae, modelchecking01, modelchecking02}.
This process facilitates the reduction of problems associated with either an infinite state system or a finite but excessively large one.
Although abstraction mitigates this issue to a certain degree, challenges persist when addressing large systems, often hindering the execution of some model checking experiments.
Some abstraction techniques have been introduced for use in Maude \cite{maude}, such as equation abstractions \cite{jose1} and predicate abstractions \cite{jose2}.
They target state space explosion but involve manual overhead.
While equational abstractions require the manual discovery of equations and often rely on theorem proving to guarantee property preservation in conditional equations,
and predicate abstractions depend on E-unification, which requires the strict finite variant property, and manual adjustment of state predicates to resolve spurious counterexamples,
our methodology offers a more accessible and streamlined verification process.
Although the construction of the abstract model and simulation relations remains a manual task, the validity of their relationship is verified semi-automatically using standard Maude commands (specifically \texttt{srew}, \texttt{red}, and \texttt{search}).
This approach reduces the reliance on extensive deductive expertise by offloading the verification of the abstraction's correctness to an executable model-checking-based procedure.

{\'{A}}kos and Zolt{\'{a}}n have proposed a framework that extends abstraction and refinement methods to improve the efficiency of CEGAR \cite{hajdu}.
This framework mitigates the state space explosion problem.
It is also noted that even with the framework, sufficient effectiveness may not be achieved in certain cases.
In the paper of Mann et al. \cite{mann}, the authors extend CEGAR for the model checking of infinite-state systems containing arrays through the incorporation of prophecy variables.
By automatically adding auxiliary variables (namely, history and prophecy variables) to the system, their approach transforms proofs that traditionally require universally quantified invariants into quantifier-free forms.
It has been reported, however, that in some cases the CEGAR loop may fail to terminate or suffer from scalability issues.
The primary distinction between CEGAR-based approaches and our methodology lies in the nature of the abstraction process. While CEGAR is a reactive technique that refines models based on spurious counterexamples, our approach is proactive, leveraging domain knowledge to construct abstract models such as SCP and ABP.
This manual but strategic abstraction provides enhanced transparency and interpretability; specifically, the simulation relations maintain a human-understandable mapping between concrete and abstract behaviors, a feature often lacking in automated refinements.
Furthermore, our methodology effectively addresses the scalability limitations of CEGAR by offloading the verification of complex protocols to a standardized sequence of symbolic Maude commands.
These commands are largely independent of the parameter scales that typically trigger state space explosion.
As summarized in Table~\ref{performance_comparison}, our experimental results confirm this efficiency: Steps 2 and 3 were completed within sub-second execution times, even with large parameters that rendered direct model checking infeasible.

Fisler and Vardi explored automated bisimulation minimization as a state space reduction technique within an automata-theoretic framework \cite{fisler}.
However, in most cases, the computational overhead of performing this minimization surpasses that of conventional model checking.
They argued that it is unsuitable for symbolic model checking because the BDDs required to compute the bisimulation relation become excessively large.
Bustan and Grumberg\cite{bustan1} have conducted research on a minimization algorithm that processes a Kripke structure $M$ and returns the smallest structure that is simulation equivalent to $M$.
Additionally, they developed the Partitioning Algorithm, which constructs the quotient structure for $M$ without the need to build the simulation preorder explicitly.
The concept of simulation equivalence is positioned as being less stringent than that of bisimulation, but more stringent than that of simulation preorder.
This relation preserves ACTL and LTL.
For the algorithms introduced in the paper \cite{bustan1} to be applicable to practical large-scale systems or complex models,
further studies are required to improve the time and space complexity of the algorithms.
Jansen et al. propose an algorithm that significantly improves the efficiency of computing branching bisimilarity on Labelled Transition Systems (LTSs)\cite{jansen}.

Zhou et al. proposes an approach to addressing the state space explosion problem in System on a Programmable Chip (SOPC) architectures\cite{zhou}.
The authors introduce techniques based on Register Transfer Level (RTL) code such as variable reduction and branch relation reduction tailored to the verification properties to compress the state space.
WASIM \cite{DBLP:conf/tacas/FangZ23} is a symbolic simulation framework for hardware formal verification.
Users can define abstraction and refinement functions through its Python API to control the level of abstraction.
It has been reported that leveraging domain knowledge in this way can achieve efficient verification.
Our methodology shares this philosophy by using domain knowledge to construct appropriate abstract models (such as SCP and ABP) for complex protocols.

Bae et al. \cite{DBLP:conf/wrla/BaeELMS24} proposed a deductive model checking methodology that synergistically integrates narrowing-based logical model checking with inductive theorem proving via the NuITP tool \cite{DBLP:conf/ppdp/Duran00S24}.
Their approach, implemented in the DM-Check tool, mitigates verification effort by employing an inductive theorem prover as an oracle to discharge verification conditions, such as subsumption checks.
While their case studies demonstrate the verification of ABP, our methodology extends to the more complex SWP.
SWP involves more intricate control logic than ABP, including window management and circular sequence numbers.
While their method excels in handling complex logical constraints and conditional rewrite theories, the proposed approach focuses on addressing the state space explosion caused by large protocol parameters, such as window sizes in SWP, through hierarchical simulation compositions.

While the theory of simulation has long been established in formal methods, our work provides a practical and engineering contribution by redefining this classical theory into a highly efficient and human-understandable ``procedure'' within the Maude model checking environment.
By operationalizing simulation-based verification through a standardized sequence of Maude commands, we bridge the gap between abstract theoretical relations and executable verification workflows.
This procedural approach enhances the accessibility of protocol verification, making it more intuitive for practitioners to manage compared to purely deductive or symbolic methods that often require intricate manual proofs or complex logical constraints.

Hubert and Fr{\'{e}}d{\'{e}}ric provide a comprehensive review of the state of equivalence checking, which uses relations such as bisimulation and simulation, focusing on its development over a span of 40 years\cite{hubert}.
They trace the evolution of algorithms, methodologies, and software tools that have been developed to address the challenges of equivalence checking in various computing paradigms, including concurrent, probabilistic, stochastic, mobile, and real-time systems.
They emphasize that equivalence checking has the potential to enhance the overall effectiveness of the verification process when closely integrated with model checking and mention that it is expected to play a growing role in the future.
Given this context, combining model checking with simulation has been explored as a promising approach.
In this paper, our proposed technique is one of the first approaches that combine model checking with simulation.
Our approach provides a unified framework within Maude for formal specification, model checking, and the verification of simulation relations.

\section{Conclusion}\label{sec:conclusion}
We have proposed a new technique to mitigate the state space explosion in model checking by combining it with simulation relations.
Our methodology demonstrates that a concrete state machine $M$ can be verified by showing it is simulated by a more abstract state machine $M_A$ that satisfies a corresponding invariant property.
By operationalizing this classical theory into a standardized sequence of Maude commands---specifically \texttt{srew}, \texttt{red}, and \texttt{search}---we have bridged the gap between abstract theoretical relations and executable verification workflows.

We also proposed a technique for the composition of simulation relations, enabling hierarchical abstraction to handle even larger and more complex systems.
The case studies on the Alternating Bit Protocol and the Sliding Window Protocol demonstrate that our approach can formally verify invariant properties where direct model checking is practically infeasible due to state space explosion.
A key advantage of this technique is the reusability of existing formal specifications as abstract models, significantly reducing the verification effort for new systems.
To ensure rigor and efficiency, we developed a support tool that automates case splitting and command execution, minimizing human error.

While this paper focuses on invariant properties, extending our technique to handle liveness properties is an important direction for future research.
Additionally, although the construction of abstract models and the identification of simulation relations currently require manual effort, we aim to develop methodologies to assist users in discovering effective simulation candidates.
Further case studies will be conducted to demonstrate the versatility of our unified framework within Maude for a broader range of complex distributed systems.

\bibliographystyle{unsrt}
\bibliography{reference}

\end{document}